\documentclass[journal]{IEEEtran}

\usepackage{rotating}
\usepackage{chngcntr}
\usepackage{apptools}
\usepackage{listings}
\AtAppendix{\counterwithin{thm}{section}}
\usepackage{graphicx}
\usepackage{graphics}
\usepackage{amssymb}
\usepackage{amsmath}
\usepackage{color}
\usepackage{soul}
\usepackage{cite}
\usepackage{xspace}
\usepackage{algpseudocode}
\usepackage{bbm}
\usepackage{comment}
\usepackage{algorithmicx}
\usepackage{subfig}
\usepackage{psfrag}
\usepackage{tikz}
\usetikzlibrary{shapes,arrows}
\usetikzlibrary{positioning}

\tikzstyle{block}=[draw opacity=0.7,line width=1.4cm]

\definecolor{CranJ}{cmyk}{0,0.69,0.54,0.04} 
\definecolor{PinkJ}{cmyk}{0,0.71,0.43,0.12} 
\definecolor{Cran}{cmyk}{0,0.73,0.41,0.29} 
\definecolor{VRed}{cmyk}{0,0.75,0.25,0.2} 
\definecolor{ORed}{cmyk}{0,0.75,0.75,0} 
\definecolor{CBlue}{cmyk}{1,0.25,0,0} 

\title{\LARGE \bf Dynamic Average Consensus in the Presence of Communication Delay Over Directed Graph Topologies}

\author{Hossein Moradian and Solmaz S. Kia %
  \thanks{The authors are with the Department of Mechanical and Aerospace Engineering, University of California Irvine, Irvine, CA 92697,  
    {\tt\small \{hmoradia,solmaz\}@uci.edu}}%
}

\newcommand{\ee}{\operatorname{e}}
\newcommand{\conj}{\operatorname{conj}}
\newcommand{\ii}{\operatorname{i}}
\newcommand{\intor}{\operatorname{int}}
\newcommand{\extor}{\operatorname{ext}}

\newcommand{\real}{{\mathbb{R}}}

\newcommand{\integer}{{\mathbb{Z}}}

\newcommand{\complex}{{\mathbb{C}}} 
\newcommand{\realpositive}{{\mathbb{R}}_{>0}}

\newcommand{\rank}{\operatorname{rank}}

\newcommand{\re}[1]{\operatorname{Re}(#1)}
\newcommand{\im}[1]{\operatorname{Im}(#1)}

\newcommand{\until}[1]{\in\{1,\dots,#1\}}

\renewcommand{\baselinestretch}{.985}

\newcommand{\vect}[1]{\boldsymbol{\mathbf{#1}}}
\newcommand{\vectsf}[1]{\vect{\mathsf{#1}}}

\newcommand{\dvect}[1]{\dot{\vect{#1}}}

\newcommand{\solmaz}[1]{{\color{red}#1}}

\newtheorem{thm}{Theorem}[section]

\newtheorem{rem}{Remark}[section]
\newtheorem{cor}{Corollary}[section]
\newtheorem{lem}{Lemma}[section]

\newcommand{\oprocendsymbol}{\hbox{$\bullet$}}
\newcommand{\oprocend}{\relax\ifmmode\else\unskip\hfill\fi\oprocendsymbol}

\makeatletter
\renewcommand*{\@opargbegintheorem}[3]{\trivlist
      \item[\hskip \labelsep{ #1\ #2}] (#3):\ \itshape}
\makeatother

\ifCLASSINFOpdf
\else
\fi

\allowdisplaybreaks
\begin{document}
\title{{\huge On the Positive Effect of Delay on the Rate of Convergence of a Class of Linear Time-Delayed Systems}}

\author{Hossein Moradian and Solmaz S. Kia
  \thanks{The authors are with the Department of Mechanical and Aerospace Engineering, University of California Irvine, Irvine, CA 92697,  
    {\tt\small \{hmoradia,solmaz\}@uci.edu}. This work is supported by NSF CAREER grant ECCS 1653838.  }%
}

%
%

\markboth{}%
{}

\maketitle

\begin{abstract}
 This paper is a comprehensive study of a long observed phenomenon of increase in the stability margin and so the rate of convergence of a class of linear systems due to time delay. 
 We use Lambert W function to determine (a) in what systems the delay can lead to increase in the rate of convergence, (b) the exact range of time delay for which the rate of convergence is greater than that of the delay free system, and (c) an estimate on the value of the delay that leads to the maximum rate of convergence. 
 For the special case when the system matrix eigenvalues are all negative real numbers, we expand our results to show that the rate of convergence in the presence of delay depends only on the eigenvalues with minimum and maximum real parts. Moreover, we determine the exact value of the maximum rate of convergence and the corresponding maximizing time delay. 
We demonstrate our results through a numerical example on the practical application in accelerating an agreement algorithm for networked~systems by use of a delayed feedback.

\end{abstract}

\IEEEpeerreviewmaketitle

\section{Introduction}
\vspace{-0.05in}
In this paper, we study the effect of a fixed \emph{time delay} $\tau\in\realpositive$ on the rate of convergence of the retarded time-delayed system
\vspace{-0.2in}
\begin{subequations}\label{eq::DDE-sys}
\begin{align}
\dvect{x}(t)&=\vect{A}\vect{x}(t-\tau),\\
\vect{x}(t) &=\vect{\phi}(t),\,\,\,  ~ t\in [-\tau,0],
\end{align}
\end{subequations}
where $\vect{x}(t)\in\real^n$ is the state variable at time $t$, $\vect{\phi}(t)$ is~a specified pre-shape function and $\vect{A}$ is a Hurwitz matrix. 
For this system, \emph{the continuity stability property} theorem for linear time-delayed systems~\cite[Proposition 3.1]{SN:01}  guarantees the existence of the connected \emph{admissible} range of delay, $[0,\bar{\tau})\subset\real{>0}$, for which the exponential stability is preserved.   
 Moreover, the \emph{critical} value of delay $\bar{\tau}>0$, beyond which the system is unstable, is the smallest value of the time delay for which the rightmost root 
 (RMR) of the characteristic equation (CE) of~\eqref{eq::DDE-sys} is on the imaginary axis for the first time. However, as shown in Fig.~\ref{fig::cont_ex}, for some systems the RMR of the CE is not necessarily traversing monotonically towards the right half complex plane as $\tau$ increases. For those systems, contrary to intuition, for certain ranges of delay the rate of convergence is greater than the delay free case (recall that the exact value of the worst convergence rate of system~\eqref{eq::DDE-sys} is determined by the magnitude of the real part of the RMR of its 
CE~\cite{TH-ZL-YS:03,SD-JN-AGU:11}). 

For system~\eqref{eq::DDE-sys}, when $\tau\!=\!0$, the RMR of the CE is the rightmost eigenvalue of $\vect{A}$, and when $\tau\!>\!0$, it can be specified by use of the Lambert W function~\cite{SY-PWN-AGU:07a}. There are also other methods to estimate the rate of convergence of system~\eqref{eq::DDE-sys}~\cite{BL-KS:94,VNP-PN:06,SY-PWN-AGU:07b,WK-MC-DM-MK:09,DB-SM-RV:05}. Despite abundance of literature on determining the convergence rate of linear systems for a given amount of time delay~\cite{SY-PWN-AGU:07b,DB-SM-RV:05,NDH:50,SN:01,BL-KS:94,VNP-PN:06,WK-MC-DM-MK:09,TI-TE-GO:17}, there are very few results that address how the rate of convergence varies with time delay. In this paper, we aim to use analytical analysis of the variation of rate of convergence vs. time delay to investigate (a) in what type of system~\eqref{eq::DDE-sys} the delay can lead to increase in the rate of convergence, (b) the exact range of time delay for which the rate of convergence is greater than that of the delay free system, and (c) an estimate on the value of the delay that leads to the maximum rate of convergence. This study extends our fundamental understanding of the internal dynamics of linear time-delay systems, and is useful in identifying rules that facilitate design of systems with fast response and improved stability margin in the presence of non-zero time delay.  A practical application of our results is in design of accelerated form of the average consensus (agreement) algorithms~\cite{ros-jaf-rmm:07} in network systems, which we demonstrate in Section V. Agreement algorithms in network systems play a crucial role in facilitating many cooperative tasks (see~\cite{ros-jaf-rmm:07,SSK-BVS-JC-RAF-KML-SM:19} for examples), and their fast convergence is always desired. 

\begin{figure}[t]
  \begin{center}
\includegraphics[trim={12pt 10pt 0 0},scale=0.55]{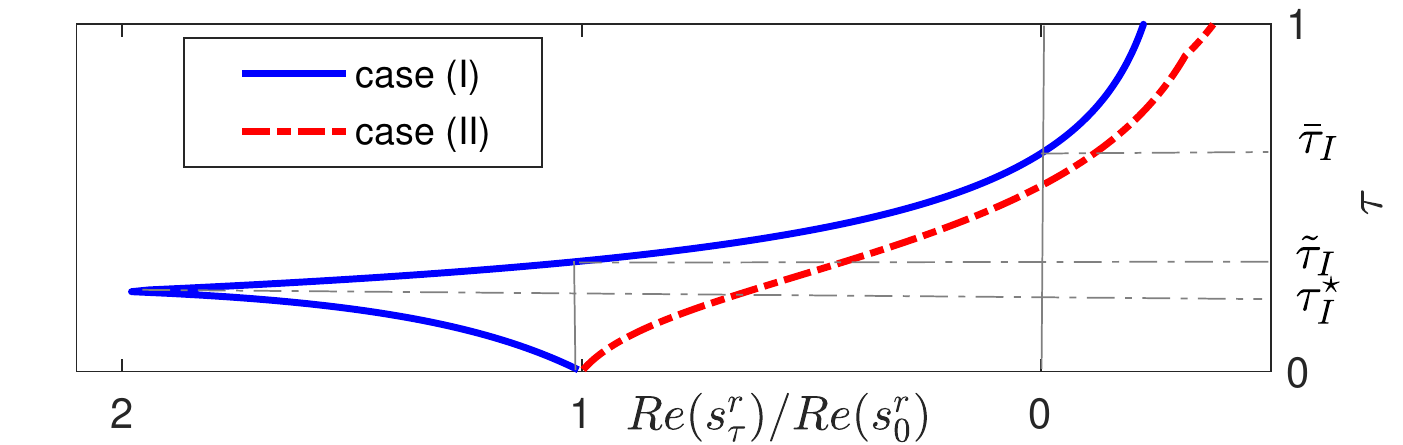}
\end{center}
  \caption{\small{
  The normalized real part of the RMR $s_\tau^r$ of the CE for the Laplacian dynamics corresponding to Case (I) and Case (II) in Section~\ref{sec::num_ex} vs. $\tau\!\in\![0,1]$. Since $\re{s_\tau^r}/\re{s_0^r}\!=\!\rho_\tau/\rho_0$, we can see that for the system in Case (I) the rate of convergence can increase with time delay but this is not the case for the system in Case (II).}
  }
 \label{fig::cont_ex}\vspace{-0.2in}
\end{figure}

Increase of stability margin and the rate of convergence of linear systems with delay has been observed in the literature~\cite{BG-SM-MHS:98,MC-DAS-EMY:06,YC-WR:10,ZM-YC-WR:10,WQ-RS:13}. However, the mathematics behind this phenomenon is not fully understood. This is due to the technical challenges that emanates from the fact that the CE of linear time delayed systems is transcendental and have an infinite number of roots in the complex plane. For system~\eqref{eq::DDE-sys},~\cite{WQ-RS:13} offers a set of interesting insights into the problem. The main result of~\cite{WQ-RS:13} states that if \emph{all} the eigenvalues of $\vect{A}$ are stable and have in magnitude larger real part than imaginary part (argument of all the eigenvalues of $\vect{A}$ are in $(\frac{3\pi}{4},\frac{5\pi}{4})$), the magnitude of the real part of the RMR of the CE and consequently the convergence rate of system~\eqref{eq::DDE-sys} increases with delay. Also,~\cite{WQ-RS:13} shows that when all the eigenvalues of Hurwitz matrix $\vect{A}$ are real, the ultimate bound on the maximum achievable rate of convergence via time delay is $\ee\approx2.71828$ times the delay free rate. The analysis method of~\cite{WQ-RS:13} relies on factoring the CE of system~\eqref{eq::DDE-sys} into $n$ scalar equations that each depends on one of the eigenvalues of $\vect{A}$, and studying how the roots of these $n$ scalar equations are affected by $\tau>0$. However, since the relative size of rate of change of the real part of the RMR of each scalar equation with delay is not known, the variation of the real part of the RMR of the CE with delay can not be determined fully.

In this paper, we take a different approach to analyze how the real part of the RMR of the CE of~\eqref{eq::DDE-sys} changes with delay. In our work, we take advantage of the fact that the exact location of the RMR of the CE of~\eqref{eq::DDE-sys} is given in a closed-form by an explicit function defined by the Lambert W function. By analyzing the rate of change of this function, we are able to recover the results in~\cite{WQ-RS:13} and extend the facts known about the variation of convergence rate of system~\eqref{eq::DDE-sys} with delay in the following directions. First, we show that the
 rate of convergence of~\eqref{eq::DDE-sys} can increase with time delay if and only if the argument of \emph{only} the rightmost eigenvalue(s) of the system matrix is strictly between $\frac{3\pi}{4}$ and $\frac{5\pi}{4}$; if a system does not satisfy this condition, its rate of convergence is in fact decreases strictly with delay in its admissible delay range. We then proceed to determine the exact range of time delay for which a system has a rate of convergence greater than the rate of a delay free system. For delays beyond this range we show that the rate of convergence decreases strictly. Our next result is to obtain an estimate on the value of the time delay corresponding to maximum achievable rate. For the special case when the system matrix eigenvalues are all negative real numbers, the relative ease in mathematical manipulations allows us also to expand our results to show that the rate of convergence in the presence of delay depends only on the eigenvalues with minimum and maximum real parts. Moreover, we determine the exact value of the maximum rate of convergence and the corresponding maximizing time delay.
A preliminary version of this paper, which discusses  the case of $\vect{A}$ having real eigenvalues appeared in~\cite{HM-SSK:18b}.

\emph{Notation}: 
We let $\complex^l=\{x\in\complex|\re{x}\leq0\}$ and $\complex_{-}^l=\{x\in\complex|\re{x}<0\}$.
The set of eigenvalues of matrix $\vect{A}\in\real^{n\times n}$ is $\text{eig}(\vect{A})$. For sets $\mathcal{A}$ and $\mathcal{B}$, $\mathcal{A}\subsetneq \mathcal{B}$ means that $\mathcal{A}$ is the strict subset of $\mathcal{B}$, and $\mathcal{A}\backslash\mathcal{B}$ is the set of all elements of $\mathcal{A}$ that are not elements of $\mathcal{B}$. For $\vect{s}\in\real^n$, $\text{diag}(\vect{s})$ is the diagonal matrix whose $n$ diagonal entries are the $n$ elements of the vector $\vect{s}$.

\section{A Review of Properties of Lambert W function}\label{sec::prelim}
\vspace{-0.05in}
The Lambert W function has been used to in stability analysis, eigenvalue assignment and obtaining the rate of convergence of linear time-delayed systems~\cite{SY-PWN-AGU:07a,HS-TM:06,SY-PWN-AGU:10,HM-SSK:19}.
For a $z\in\complex$, the Lambert $W$ function is defined as the solution of $s\,\ee^{s}=z$, i.e., $s=W(z)$. Except for $z=0$ for which $W(0)=0$, $W$ is a multivalued function with the infinite number of solutions denoted by $W_k(z)$, $k\in\integer$, where $W_k$ is called the $k^{\text{th}}$ branch of $W$ function. $W_k(z)$ can readily be computed in Matlab or Mathematica. Zero branch of the Lambert function, $W_0$ is of special interest in this paper,
which has the following properties (see~\cite{HS-TM:06,RMC-GHG-DEGH-DJJ-DEK:96,CH-YCC:05})
\begin{subequations}\label{eq::W_0_facts}
\begin{alignat}{2}
&W_0(-\frac{1}{\ee})\!=\!-1,~ W_0(0)\!=\!0,&~\\
&\re{W_0(z)}>-1, & z\in\real\backslash\{-\frac{1}{\ee}\},
\label{eq::W_0_facts_c}\\
&W_0(z)\in\real,&z\in[-\frac{1}{\ee},\infty),  \label{eq::W_0_facts_a} \\
&\im{W_0(z)}\in(-\pi,\pi)\backslash\{0\}, \,&z\in\complex\backslash(-\infty,-\frac{1}{\ee}),\label{eq::W_0_facts_b}\\
&W_0(\conj(z))=\conj(W_0(z)),  & z\in\complex.\label{eq::W_0_conjugate_prop}
\end{alignat}
\end{subequations}

Other properties of the Lambert $W$ function that we use are  
\begin{subequations}
\begin{align}
&\frac{\text{d}\,W(z)}{\text{d}\,z}=\frac{1}{z+\ee^{W(z)}},\quad z\in\complex\backslash\{\frac{1}{\ee}\},\label{eq::lambert_derivative}
\\
 &       \lim_{z\to 0}\frac{W(z)}{z}=1, \label{eq::limWz-z}
\end{align}
\end{subequations}


\section{Preliminaries and objective statement}\label{sec::Prob_formu}\vspace{-0.05in}
Let the eigenvalues of $\vect{A}$ in~\eqref{eq::DDE-sys} be $\{\alpha_1,\cdots,\alpha_n\}\subset \complex^{l}_{-}$, which are ordered according to  $|\re{\alpha_1}|\leq |\re{\alpha_2}|. \leq\dots\leq  |\re{\alpha_n}|$. The critical value of delay $\bar{\tau}\in\real_{>0}$ beyond which the system~\eqref{eq::DDE-sys} is unstable is specified as follows.
\begin{lem}[Admissible range for delay $\tau$ for the system~\eqref{eq::DDE-sys}~\cite{MB:87}] \label{thm::admis_tau}
The time-delayed system~\eqref{eq::DDE-sys} 
is exponentially stable if and only if 
 $\tau\in[0,\bar{\tau})\subset\real_{\geq0}$ where 
\begin{align}\label{eq::delay_bound}
\bar{\tau}=\min\{\,\bar{\tau}_i\}_{i=1}^n,\quad \bar{\tau}_i&=\big|\text{\rm{atan}}({\re{\alpha_i}}\big/{\im{\alpha_i}})\big|\Big/|\alpha_i|.
\end{align}
\end{lem}

As mentioned earlier, the rate of convergence of system~\eqref{eq::DDE-sys}  is determined by the magnitude of the real part of the RMR of its CE. Therefore, in the absence of the delay, the rate of convergence of system~\eqref{eq::DDE-sys} is  $\rho_0=|\re{\alpha_1}|$. As shown in~\cite{SD-JN-AGU:11}, when $\tau>0$ the magnitude of the real part of the RMR and as a result the rate of convergence of~\eqref{eq::DDE-sys} is given by
\begin{align}\label{eq::rho}
        \rho_{\tau}=\min\big\{\,\,-\frac{1}{\tau}\re{W_0(\alpha_i\tau)}\,\,\big\}_{i=1}^n.
\end{align}
Using properties of the Lambert W function, we can show $\rho_\tau$ is a continuous function of $\tau$. (see Lemma~\ref{lem::p_tau_continuous} in~the appendix).
Our objective  is to show that for  system~\eqref{eq::DDE-sys}, 
 it is possible to have $\rho_\tau>\rho_0= |\re{\alpha_1}|$  for certain values of delay $\tau\in(0,\bar{\tau})$. 
In particular, we carefully examine the variation of $\rho_\tau$ with~$\tau\in(0,\bar{\tau})$ to address the following questions: (a) for what systems delay can lead to a higher rate of convergence, (b) for what values of delay  $\rho_\tau>\rho_0=|\re{\alpha_1}|$ (c) what is the maximum value of $\rho_\tau$ and the corresponding maximizer $\tau$.

To compare the rate of convergence~\eqref{eq::rho} to the delay free rate, we define the \emph{delay rate gain function} as follows
\begin{align}\label{eq::delay_rate_gain}
       &g(x)=\begin{cases}\frac{\re{W_0(x)}}{\re{x}},&x\in\complex_{-}^l,\\
       1,& x=0.
       \end{cases}
\end{align}
For any $\alpha\in\complex_{-}^l$ using delay rate gain we can write
\begin{align}\label{eq::gain_def}
        -\frac{1}{\tau}\re{W_0(\alpha\tau)}=g(\alpha\tau)\,|\re{\alpha}|,\quad \tau\in\real_{>0}.
\end{align}
Therefore, the rate of convergence~\eqref{eq::rho} of the system~\eqref{eq::DDE-sys} can be expressed also as  \begin{align}\label{eq::rho-tau-g}
    \rho_{\tau}=\min\big\{\,g(\alpha_i\tau)\,|\re{\alpha_i}|\,\big\}_{i=1}^n,\quad \tau\in\real_{>0}.
\end{align} 

Next, we study the properties of the delay rate gain function~\eqref{eq::delay_rate_gain} to identify the ranges of $\tau$ that $g(\alpha\tau)>1$ for a given $\alpha\in\complex_{-}^l$. The proof of these auxiliary results are given in the appendix.
\begin{lem}[characterizing the solutions of $g(\alpha\tau)=0$ and $g(\alpha\tau)=1$]\label{lem::g0-g1}
For any $\alpha\in\complex_{-}^l$, the delay rate gain~\eqref{eq::delay_rate_gain} satisfies
\vspace{-0.15in}
\begin{subequations}\label{eq::g0_bar_tau}
\begin{align}
&\lim_{\tau\to0}g(\alpha\,\tau)=1,\label{eq::lim-g-0}\\
& \{\tau\in\real_{\geq0}\,|\,g(\alpha\tau)\!=\!0\}\!=\!\{\bar{\tau}\},~\bar{\tau}\!=\!   \big|\text{\rm{atan}}(\frac{\re{\alpha}}{\im{\alpha}})\big|\!\Big/\!|\alpha|,\label{eq::tau_bar}
        \end{align}  
       \end{subequations}
If $\alpha\in\real_{<0}$, then $\bar{\tau}=\pi/2|\alpha|$. For $\alpha\in\real_{<0}$ we also have 
        \begin{align}
        &\{\tau\in\real_{>0}\,|\,g(\alpha\tau)=1\}=\left\{\tilde{\tau}\right\},~ ~ \tilde{\tau}=\tilde{\theta}\cot(\tilde{\theta})/|\alpha|,\label{eq::tau_tilde_real}
        \end{align}
 where $\tilde{\theta}$ is the unique solution of $\ee^{-\theta\cot(\theta)}\!=\!\cos(\theta)$ in $(0,\pi)$, which approximately is $1.01125$.
\end{lem}
\begin{lem}[$g(\alpha\tau)$ is a continuous function of $\tau$]\label{lem::g_continuous} 
For a given $\alpha\in\complex_{-}^l$, $g(\alpha\tau)$ is a continuous function of $\tau\in\real_\geq0$. 
\end{lem}
Our next result specifies how the sign of $g(\alpha\tau)$ for a given $\alpha\in\complex_{-}^l$ changes with respect to $\tau\in\real_{>0}$. 
\begin{lem}[values of $\tau$ for which $g(\alpha\tau)>0$]\label{lem::g_0}
For a given $\alpha\in\complex_{-}^l$,
$ g(\alpha\tau)>0$ for $\tau\in(0,\bar{\tau})$, $ g(\alpha\tau)<0$ for $\tau\in(\bar{\tau},\infty)$ and $g(\alpha\bar{\tau})=0$, where $\bar{\tau}$ is given in~\eqref{eq::tau_bar}.
\end{lem}
  The proof of Lemma~\ref{lem::g_0} is given in the appendix.
 \begin{figure}
    \centering
    \includegraphics[scale=0.33]{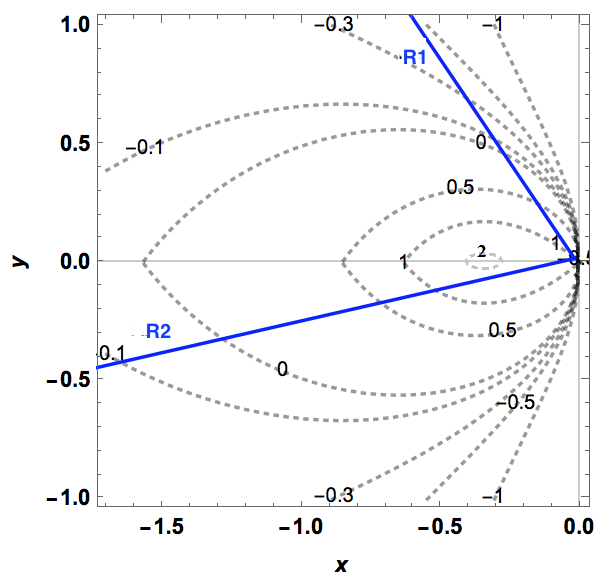}\qquad \includegraphics[trim=5pt 0 10pt 0,clip,scale=0.42,angle=0
        ]{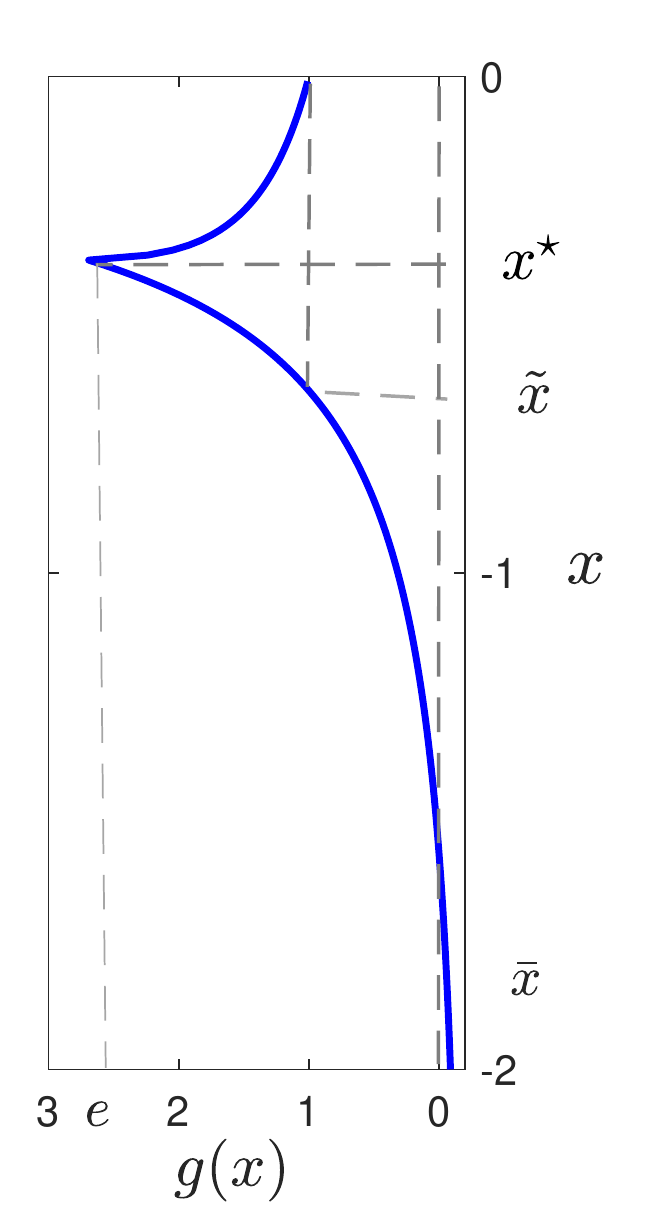}\vspace{-0.05in}  
   \caption{\small
   The plot on the right: $g(x)$ for $x\in\real_{<0}$. The plot on the left: counter plot of $g(\mathsf{x}+\mathsf{y}{\ii})$ for different values of $(\mathsf{x},\mathsf{y})\in\real_{\leq0}\times\real$. The lines R1 and R2, atop of the contour plots, show $\re{\alpha\,\tau}$ versus $\im{\alpha\,\tau}$, for different values of $\tau\in\real_{\geq0}$ for, respectively,  $\alpha=-3+5{\ii}$    ($\arg(\alpha)=\frac{2\pi}{3}$) and $\alpha=-5-1.34{\ii}$ ($\arg(\alpha)=\frac{13\pi}{12}$).
     }\vspace{-0.15in}
    \label{Fig::g_countor}
\end{figure}

 
In the remainder of this section our objective is to characterize the variation of $g(\alpha\tau)$ from $1$ at $\tau=0$ to $0$ at $\tau=\bar{\tau}$ (recall~\eqref{eq::g0_bar_tau}), with the specific objective of identifying the values of $\tau\in(0,\bar{\tau})$ for which $g(\alpha\tau)>1$. In this regard, first we consider the plot on the right hand side of Fig.~\ref{Fig::g_countor}, which shows variation of $g(x)$ versus $x\in\real_{\leq 0}$ over $x\in[0,-2]$. This plot reveals a set of interesting facts as follows.
\begin{itemize}
    \item At $x\!=\!\bar{x}\!=-\!\frac{\pi}{2}$ we have  $g(\bar{x})=0$, while for any $x\in(\bar{x},0)$, $g(x)>0$ and for any $x<\bar{x}$, $g(x)<0$ (as expected according to Lemma~\ref{lem::g_0}).
    \item At  $x=x^\star=-\frac{1}{\ee}$, the maximum delay rate gain $g(x^\star)=\ee$ is attained. For any $x\in(x^\star,0)$, $g(x)$ strictly decreases from $\ee$ to $1$, while for any $x\in(\bar{x},x^\star)$, $g(x)$ increases strictly from $0$ to $\ee$. 
    \item  Let $\tilde{x}$ be the non-zero solution of $g(\tilde{x})=1$, $x\in\real_{<0}$ (an approximate value of $\tilde{x}$ is $-0.63336$, see Lemma~\ref{lem::g0-g1} for analytic characterization of $\tilde{x}$ using $\tilde{x}=\alpha\tilde{\tau}$). Then, for any $x\in(\tilde{x},0)$ we have $g(x)>
    1$. Also, for any $x\in[\bar{x},\tilde{x}]$, $g(x)$ strictly increases from $0$ to $1$. \end{itemize}

For a given  $\alpha\in\real_{<0}$, using the aforementioned observations, one can describe the variation of  $g(\alpha\tau)$ over $\tau\in[0,\bar{\tau}]$. We formalize this characterization in the following lemma, whose rigorous proof is given in the appendix. 
\begin{lem}[variation of $g(\alpha\tau)$ for an $\alpha\in\real_{<0}$ and $\tau\in\real_{>0}$]\label{lem::alpha-real-positiverate}
Consider the delay rate gain function~\eqref{eq::delay_rate_gain}. Let $\alpha\in\real_{<0}$ be given. Recall $\bar{\tau}$ and $\tilde{\tau}$ from Lemma~\ref{lem::g0-g1}. Let {\rm$\tau^\star=\frac{1}{\ee|\alpha|}$}. Then, the followings hold.
\begin{itemize}
   \item[(a)]  For any $\tau\in(0,{\tau}^\star)\subset(0,\bar{\tau})$, $g(\alpha\tau)>1$, and $g(\alpha\tau)$ strictly increases from $1$ to {\rm$\ee$};  {\rm$g(\alpha{\tau}^\star)=\ee$}; and for any $\tau\in({\tau}^\star,\bar{\tau})\subset(0,\bar{\tau})$, $g(\alpha\tau)>0$, and $g(\alpha\tau)$ strictly decreases from {\rm$\ee$} to $0$.
   \item[(b)] For any $\tau\in(0,\tilde{\tau})\subset(0,\bar{\tau})$, $g(\alpha\tau)>1$; $g(\alpha\tilde{\tau})=1$; $g(\alpha\bar{\tau})=0$;  for any  $\tau\in(\tilde{\tau},\bar{\tau})$, $0<g(\alpha\tau)<1$; and for $\tau\in(\bar{\tau},\infty)$, $g(\alpha\tau)<0$. 
        \item[(c)] The maximum value of $g(\alpha\tau)$ is {\rm$\ee$}, which is attained at $\tau=\tau^\star\in(0,\tilde{\tau})$. 
\end{itemize}
\end{lem}


Let the level set and the superlevel set of delay rate gain function for $c\in\real$ be respectively
\begin{subequations}
 \begin{align}
\mathcal{C}_c&=\{(\mathsf{x},\mathsf{y})\in\real_{<0}\times\real|~g(\mathsf{x}+\mathsf{y}{\ii})= c\},\label{eq::C_c} \\
\mathcal{S}_{c}&=\{(\mathsf{x},\mathsf{y})\in\real_{<0}\times\real|~g(\mathsf{x}+\mathsf{y}{\ii})\,\geq\, c\}.\label{eq::S_c}
\end{align}
\end{subequations}
As seen in the contour plots in Fig.~\ref{Fig::g_countor}, $g(x)$ attains a value greater than one for some $x=\mathsf{x}+{\ii}\,\mathsf{y}$. For example, consider the points on the lines R1 and R2 on Fig.~\ref{Fig::g_countor}, which depict, respectively,  $(\re{\alpha_{\text{R1}}\tau},\im{\alpha_{\text{R1}}\,\tau})$ and $(\re{\alpha_{\text{R2}}\tau},\im{\alpha_{\text{R2}}\,\tau})$ for $\tau\in\real_{\geq0}$, with $\alpha_{\text{R1}}\!=\!-3+5{\ii}$ and $\alpha_{\text{R2}}\!=\!-5-1.34{\ii}$. As seen, $g(\alpha_{\text{R1}}\tau)$, $\tau\in\real_{>0}$ is always less than one, while  $g(\alpha_{\text{R2}}\tau)$  is greater than one for $\tau\!\in\!(0,\tilde{\tau})$, where $\tilde{\tau}$ satisfies $g(\alpha\tilde{\tau})\!=\!1$, see also Fig.~\ref{Fig::stab_complex}. 
Therefore, we expect that a delay rate gain of greater than one is only possible for certain values of $\alpha\in\complex_{-}^l$.
In what follows, we set off to address (a) for what values of $\alpha\in\complex_{-}^l$, $g(\alpha\tau)$ can have a value greater than 1 for a $\tau\in\real_{\geq0}$, (b) what values of $\tau\in(0,\bar{\tau})$ correspond to $g(\alpha\tau)\!>\!1$, and (c) what the maximum gain $g(\alpha \tau^\star)$ and the corresponding $\tau^\star$ are.  We start our study by the following result that for a given $\alpha\in\complex_-^l$, characterizes the sign of ${\text{d} g(\alpha\tau)}/{\text{d}\tau}$ for any $\tau\in(0,\bar{\tau})$. 

\begin{lem}[variation of ${\text{d} g(\alpha\tau)}/{\text{d}\tau}$ with $\tau\in(0,\bar{\tau}{]}$ for a $\alpha\in\complex_{-}^l$]
\label{lem::desired_region_complex}
Consider the delay rate gain function~\eqref{eq::delay_rate_gain}. Let $\alpha\!\in\!\complex_{-}^l$ be given. Recall $\bar{\tau}$ from Lemma~\ref{lem::g0-g1} and $\mathcal{S}_0$ in~\eqref{eq::S_c}.~Let
\begin{align}\label{thm::desired_region_complex}
\Lambda\!=&\Big\{\!(\mathsf{x},\mathsf{y})\!\in\!\real_{\leq0}\!\times\!\real\Big|\mathsf{x}\!+\!\mathsf{y}\ii=r\ee^{\phi\ii},\,r\!=\!-\frac{\cos(2\theta)}{\cos(\theta)}\ee^{-\cos(2\theta)}\!,\nonumber\\ & ~~\phi\!=\theta\!-\!\cos(2\theta)\tan(\theta)\ee^{-\cos(2\theta)}, ~\frac{3\pi}{4}\leq\!\theta\leq\!\frac{5\pi}{4}\!\Big\}.
\end{align}
Then, the following assertions hold.
\begin{itemize}
\item[(a)] For any $\tau\in(0,\bar{\tau}]$, the delay rate gain  satisfies 
 {\rm$\frac{\text{d} g(\alpha\tau)}{\text{d}\tau}>0$} if
{\rm$(\re{\alpha\tau},\im{\alpha\tau})\in\intor(\Lambda)=\big\{(\mathsf{x},\mathsf{y})\!\in\real_{\leq0}\!\times\!\real\,\big|\mathsf{x}+\mathsf{y}\ii=r\,\ee^{\phi\ii},\, 0<r<-\frac{\cos(2\theta)}{\cos(\theta)}\,\ee^{-\,\cos(2\theta)}\,\,,\,\,\phi=\theta-\cos(2\theta)\tan(\theta)\ee^{-\cos(2\theta)},\,\,,\frac{3\pi}{4}\leq\!\theta\leq\!\frac{5\pi}{4}\big\}$}, {\rm$\frac{\text{d} g(\alpha\tau)}{\text{d}\tau}=0$} if {\rm$(\re{\alpha\tau},\im{\alpha\tau})\in\Lambda\backslash\{(-\frac{1}{\ee},0)\}$}, and  {\rm$\frac{\text{d} g(\alpha\tau)}{\text{d}\tau}<0$} if
$(\re{\alpha\tau},\im{\alpha\tau})\in(\mathcal{S}_0\backslash(\intor(\Lambda)\cup\Lambda))$. 
\item[(b)] If $\arg(\alpha)\not\in(\frac{3\pi}{4},\frac{5\pi}{4})$, then 
$\frac{\text{d} g(\alpha\tau)}{\text{d}\tau}<0$ for $\tau\in(0,\bar{\tau}]$.

\item[(c)] If $\arg(\alpha)\in(\frac{3\pi}{4},\frac{5\pi}{4})$,  then 
$\frac{\text{d} g(\alpha\tau)}{\text{d}\tau}>0$ for $\tau\in(0,\tau^\star)$, and 
$\frac{\text{d} g(\alpha\tau)}{\text{d}\tau}<0$ $\tau\in(\tau^\star,\bar{\tau}]$, where $\tau^\star$ is  
\begin{align}\label{eq::tau-star-complex}
     \tau^{\star}=-\frac{\cos(2\,\theta^\star)}{|\alpha|\,\cos(\theta^\star)}\,\ee^{-\cos(2\,\theta^\star)},
     \end{align}
     in which $\theta^\star$ is the unique solution of 
     \begin{align}\label{eq::theta}
    \!\!\!\! \!\!\!\!\arg(\alpha)\!=\!\theta\!-\!\cos(2\,\theta)\tan(\theta)\ee^{-\cos(2\,\theta)},~\theta\!\in\!(\frac{3\pi}{4},\!\frac{5\pi}{4}).
     \end{align}
    At $\tau=\tau^\star$, we have $(\re{\alpha\tau^\star},\im{\alpha\tau^\star})\in\Lambda$.  If $\alpha\in\real_{<0}$, then $\tau^\star=\frac{1}{\ee|\alpha|}$ and 
\begin{align}\label{eq::drev_g_at_tau_star}
&\!\!\lim_{\tau\to\tau^{\star}\,\!^-}\!\!\!\!\frac{\text{d}g(\alpha\,\tau)}{\text{d}\tau}\!=\!+\infty,\,\,\lim_{\tau\to\tau^{\star}\,\!^+}\!\!\!\!\frac{\text{d}g(\alpha\,\tau)}{\text{d}\tau}\!=\!-\frac{5\ee^2|\alpha|}{3}.
\end{align} Lastly, if $\alpha\in\complex_{-}^l\backslash\real_{<0}$, then
     $\frac{\text{d} g(\alpha\tau)}{\text{d}\tau}=0$ at $\tau=\tau^\star$.
\end{itemize}
\end{lem}

\begin{figure}
    \centering
  \includegraphics[trim=30pt 0 40pt 0,clip,scale=0.176]{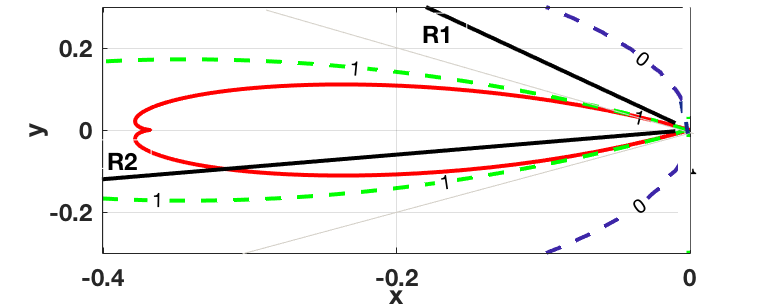}
    \includegraphics[trim=40pt 0 35pt 0,clip,scale=0.176]{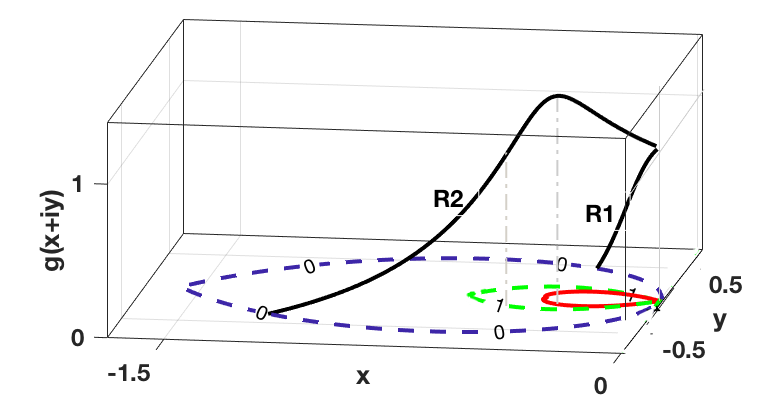}
   \caption{\small The left plot shows   $\Lambda$ in~\eqref{thm::desired_region_complex} (red curve) in a $(\mathsf{x},\mathsf{y})$ plane along with the level sets $\mathcal{C}_1$ (green curve) and $\mathcal{C}_0$ (blue curve). The right plot shows $g(\mathsf{x}+\ii\mathsf{y})$ vs. $(\mathsf{x},\mathsf{y})$ for points along the lines R1 and R2 shown in the top. 
}\vspace{-0.2in}
    \label{Fig::stab_complex}
\end{figure}

Figure~\ref{Fig::stab_complex} depicts $\Lambda$ in~\eqref{thm::desired_region_complex} (red curve) in a $(\mathsf{x},\mathsf{y})$ plane along with the level sets $\mathcal{C}_1$ (green curve) and $\mathcal{C}_0$ (blue curve). Note that $\Lambda$ is a simple closed curve that divides the space into a bounded interior and an unbounded exterior area. Moreover, $\Lambda$ is located inside $\mathcal{C}_1$ and between the lines $\mathsf{y}=\pm\mathsf{x}$. We note that $\mathsf{y}=\pm\mathsf{x}$ are also tangent to $\mathcal{C}_1$ at the origin 
(a rigorous study of these geometric observations is available in the appendix). 
The bottom plot in Fig.~\ref{Fig::stab_complex} shows also how $g(\mathsf{x}+\mathsf{y}\,\ii)$ varies along the lines R1 and R2 for points inside $\mathcal{S}_0$. As seen, the delay rate gain is strictly decreasing along R1. However, along R2 it is strictly increasing until R2 intersects $\Lambda$, and  it is strictly decreasing afterward until R2 intersects $\mathcal{C}_0$.  With Lemma~\ref{lem::desired_region_complex} at hand, we are now ready to present the main result of this section.
\begin{thm}[$g(\alpha\tau)$ vs. $\tau\in\real_{>0}$ for $\alpha\in\complex_{-}^l$]\label{lem::g_complexAlpha}
Consider the delay rate gain function~\eqref{eq::delay_rate_gain}. Let $\alpha\in\complex_{-}^l$ be given. Recall $\bar{\tau}$ from Lemma~\ref{lem::g0-g1}. Then, the following assertions hold. \begin{itemize}
 \item[(a)] If $\arg(\alpha)\notin(\frac{3\pi}{4},\frac{5\pi}{4})$, then $g(\alpha\tau)$ decreases strictly from $1$ to $0$ for $\tau\in[0,\bar{\tau}]$. 
 
\item[(b)] If $\arg(\alpha)\in(\frac{3\pi}{4},\frac{5\pi}{4})$, then: (i)
$g(\alpha\tau)$ increases strictly from $1$ to $g(\alpha\tau^\star)$ for $\tau\in[0,\tau^\star]$ and  decreases strictly from $g(\alpha\tau^\star)$ to $0$ for $\tau\in[\tau^\star,\bar{\tau}]$, where $\tau^\star$ is specified in the statement (c) of Lemma~\ref{lem::desired_region_complex}; 
(ii) 
$\tilde{\tau}\in(0,\bar{\tau})$ such that $g(\alpha\tilde{\tau})=1$ exists and is unique and satisfies $0<\tau^\star<\tilde{\tau}<\bar{\tau}$; (iii) $g(\alpha\tau)>1$ for $\tau\in(0,\tilde{\tau})$, and $ g(\alpha\tau)<1$ for $\tau\in(\tilde{\tau},\bar{\tau}]$.
     \end{itemize}
\end{thm}
\begin{proof}
We recall that $\lim_{\tau\to0}g(\alpha\tau)=1$, $g(\alpha\bar{\tau})=0$, and $g(\alpha\tau)$ is a continuous function of $\tau\in\real_{\geq0}$. Then, the proof of the statements~(a) and (b) follows respectively from  the statements (b) and (c) of Lemma~\ref{lem::desired_region_complex}. 
\end{proof}


\section{Delay effects on the rate of convergence}\label{sec::main}\vspace{-0.05in}
In this section, using the properties of the delay rate gain function, we inspect closely how the convergence rate of  system~\eqref{eq::DDE-sys} changes with time delay. 
We start by defining some notations. First, let
\begin{subequations}\label{eq::def_I}
\begin{align}
   & \mathcal{I}_1=\{k\until{n}|\re{\alpha_k}=\re{\alpha_1}\} ,\\&{\mathcal{I}_{\text{in}}}=\big\{k\until{n}|\arg(\alpha_k)\in(\frac{3\pi}{4},\frac{5\pi}{4})\big\},\\&{\mathcal{I}_{\text{out}}}=\{1,\cdots,n\}\backslash \mathcal{I}_{\text{in}}.
\end{align}
\end{subequations}
Next, recall $\bar{\tau}^i$ from~\eqref{eq::delay_bound}, and also let
\begin{subequations}
\begin{align}\label{eq::tau_ii}
    \tilde{\tau}_i&=\max\{\tau\in[0,\bar{\tau}_i]\,|\,g(\alpha_i\tau)=1\},\\
\tau_i^\star&=
\arg\max_{\tau\in[0,\bar{\tau}_i]}g(\alpha_i\tau), \end{align}
\end{subequations}
 $i\until{n}$. Recall from Lemma~\ref{lem::g0-g1} that $g(\alpha_i\bar{\tau}_i)=0$. Moreover, by virtue of Lemma~\ref{thm::admis_tau} the critical value of delay for the time-delayed system~\eqref{eq::DDE-sys} is $\bar{\tau}=\min\{\bar{\tau}_i\}_{i=1}^n$. Also from Theorem~\ref{lem::g_complexAlpha} we can obtain the following result.
 \begin{cor}[relative size of $(\bar{\tau}_i,\tau_i^\star,\tilde{\tau}_i)$]
 Consider $\textup{eig}(\vect{A})=\{\alpha_i\}_{i=1}^n\!\subset\!\complex_{-}^l$ of the system matrix of~\eqref{eq::DDE-sys}.~Then, 
  \begin{subequations}\label{eq::component_tau_i}
 \begin{alignat}{2}
    & 0<\tau^{\star}_i<\tilde{\tau}_i<\bar{\tau}_i,\quad &i\in{\mathcal{I}_{\text{in}}},\,\\
    & ~~~~\tau_i^\star=\tilde{\tau}_i=0,\quad &\,\,i\in{\mathcal{I}_{\text{out}}},
 \end{alignat}
 \end{subequations}
 where for $i\!\in\!{\mathcal{I}_{\text{in}}}$, $\tau_i^\star$ is a unique point obtained from~\eqref{eq::tau-star-complex} for $\alpha\!=\!\alpha_i$, and $\tilde{\tau}_i$ is the unique solution of $g(\alpha_i\tau)\!=\!1$ in $(0,\bar{\tau}_i)$.
 \end{cor}
 Lastly, we let
\begin{subequations}
\begin{align}
     \tilde{\tau}&=\max\{\tau\in[0,\bar{\tau}]\,|\,\rho_\tau=|\re{\alpha_1}|\},\label{eq::tilde_tau_sys}\\
     \tau^\star&=\arg\max_{\tau\in[0,\bar{\tau})}\rho_\tau,\\
     \rho^{i}_\tau&=g(\alpha_i\tau)|\re{\alpha_i}|,\quad i\in\{1,\cdots,n\}.
\end{align}
\end{subequations}
With the proper notations at hand, we now present our first result, which specifies what system~\eqref{eq::DDE-sys} can have a higher rate of convergence in the presence of the time delay.
\begin{thm}[Systems for which rate of convergence can increase by time delay]\label{thm::main_type_complex}
 Consider the linear time-delayed system~\eqref{eq::DDE-sys} when $\{\alpha_i\}_{i=1}^n\subset\complex_{-}^l$. Recall the admissible delay bound $\bar{\tau}$ given by~\eqref{eq::delay_bound}. Then there always exists a $\tau\in(0,\bar{\tau})$ for which $\rho_\tau>\rho_0=|\re{\alpha_1}|$ if and only if 
 {\rm$\mathcal{I}_1\subset\mathcal{I}_{\text{in}}$}. 
\end{thm}
\begin{proof}
If there exists a $j\in\mathcal{I}_1$ that is not in $\mathcal{I}_{\text{in}}$, i.e., $\arg(\alpha_j)\notin(\frac{3\pi}{4},\frac{5\pi}{4})$, then by virtue of Lemma~\eqref{lem::g_0} and the statement (a) of Theorem~\ref{lem::g_complexAlpha} we know that $g(\alpha_j\tau)<1$ for any $\tau\in\real_{>0}$. Subsequently, since $|\re{\alpha_1}|\leq \cdots\leq |\re{\alpha_n}|$, from the definition of the $\rho_\tau$ in~\eqref{eq::rho-tau-g} we obtain that $\rho_\tau< |\re{\alpha_1}|$ for all $\tau\in(0,\bar{\tau})$.
Now, assume that $\mathcal{I}_1\subset\mathcal{I}_{\text{in}}$. Then, by virtue of the statement (b) of Theorem~\ref{lem::g_complexAlpha}, we know that $g(\alpha_i\tau)>1$ for any $\tau\in(0,\tilde{\tau}_i)$ for $i\in\mathcal{I}_1$, (recall  $\tilde{\tau}_i\in(0,\bar{\tau}_i)$ due to~\eqref{eq::component_tau_i}). Subsequently, since $g(0)=1$ and $|\re{\alpha_k}|> |\re{\alpha_1}|$ for $k\in\{1,\cdots,n\}\backslash\mathcal{I}_1$, then by virtue of Lemma~\ref{lem::g_continuous}, there exists a $\check{\tau}\in((0,\min\{\tilde{\tau}_i\}_{i\in\mathcal{I}_{1}})\cap (0,\bar{\tau}))$ such that $g(\alpha_i\tau)>1$, $i\in\{1,\cdots,n\}$, for any $\tau\in(0,\check{\tau})$. Then, the proof of the sufficiency of the theorem statement follows from the definition of $\rho_\tau$ in~\eqref{eq::rho-tau-g} and its continuity with respect to $\tau\in\real_{>0}$.
\end{proof}
Our next result specifies for what values of time delay a system, which satisfies the necessary and sufficient condition of Theorem~\eqref{thm::main_type_complex}, experiences an increase in its rate of convergence in the presence of delay. This result also gives the value of $\tilde{\tau}$ and provides an estimate on the value of  $\tau^\star$.

\begin{thm}[Ranges of delay for which the rate of convergence of~\eqref{eq::DDE-sys} increases with delay]\label{thm::com_alpha_n_rate}
Consider the linear time-delayed system~\eqref{eq::DDE-sys} when $\{\alpha_i\}_{i=1}^n\subset\complex_{-}^l$. Recall the admissible delay bound $\bar{\tau}$ given by~\eqref{eq::delay_bound}. Suppose that {\rm$\mathcal{I}_1\subset\mathcal{I}_{\text{in}}$}. Then, the following assertions hold. 
\begin{itemize}
    \item[(a)] 
    $\tilde{\tau}=\min\{\eta_i\}_{i=1}^n$, where $\eta_i$ is the unique solution of  $g(\alpha_i\tau)=\frac{|\re{\alpha_1}|}{|\re{\alpha_i}|}$ for $\tau\in(0,\bar{\tau}_i)$. Moreover, $\min\{\tau_i^\star\}_{i=1}^n<\tilde{\tau}<\bar{\tau}$.
    \item[(b)] $\rho_{\tau}>\rho_0=|\re{\alpha_1}|$  for $\tau\in(0,\tilde{\tau})\subset(0,\bar{\tau})$, $\rho_{\tau}=\rho_0=|\re{\alpha_1}|$ at $\tau=\tilde{\tau}$ and $\rho_{\tau}< \rho_0=|\re{\alpha_1}|$ for $\tau\in(\tilde{\tau},\bar{\tau})$. Moreover, $\rho_\tau$ decreases strictly with $\tau\in(\tilde{\tau},\bar{\tau})$.
    \item[(c)] $\tau^\star\in\big([\min\{\tau^\star_i\}_{i=1}^n,\max\{\tau^\star_i\}_{i=1}^n]\cap(0,\tilde{\tau})\big)\subset(0,\bar{\tau}).$
     \end{itemize}
\end{thm}
\begin{proof}
For $j\in\mathcal{I}_{\text{out}}$, the statement (a) of Theorem~\eqref{lem::g_complexAlpha} guarantees that $g(\alpha_j\tau)$ is strictly decreasing from $1$  to $0$  for $\tau\in[0,\bar{\tau}_j]$. Thus, for  $j\in\mathcal{I}_{\text{out}}$, given the continuity of $g(\alpha_j\tau)$ in $\tau\in[0,\bar{\tau}_j]$ and $\frac{|\re{\alpha_1}|}{|\re{\alpha_j}|}<1$ (recall that $\mathcal{I}_{1}\not\subset\mathcal{I}_{\text{out}}$), $\eta_j$ is a non-zero unique value in $(0,\bar{\tau}_j)$ at which we have $g(\alpha_j\eta_j)|\re{\alpha_j}|=|\re{\alpha_1}|$. Moreover, for $j\in\mathcal{I}_{\text{out}}$, we have $g(\alpha_j\eta_j)|\re{\alpha_j}| <|\re{\alpha_1}|$ for $\tau\in(\eta_j,\bar{\tau}_j)$ and $g(\alpha_j\eta_j)|\re{\alpha_j}| >|\re{\alpha_1}|$ for $\tau\in(0,\eta_j)$. Recall also that $\tau_j^\star=0$ for $j\in\mathcal{I}_{\text{out}}$. For $j\in\mathcal{I}_{\text{in}}$, the statement (b) of Theorem~\eqref{lem::g_complexAlpha} guarantees that $g(\alpha_j\tau)$ is strictly increasing from $1$ to its maximum value $g(\alpha_j\tau^\star_j)>1$ for $\tau\in(0,\tau^\star_j)$ and it is strictly decreasing from $g(\alpha_j\tau^\star_j)>1$ to zero for $\tau\in[\tau_j^{\star},\bar{\tau}_j]$. Thus, for  $j\in\mathcal{I}_{\text{in}}$, given the continuity of $g(\alpha_j\tau)$ in $\tau\in[0,\bar{\tau}_j]$ and $\frac{|\re{\alpha_1}|}{|\re{\alpha_j}|}\leq1$ (recall that $\mathcal{I}_1\subset\mathcal{I}_{\text{in}}$), $\eta_j$ is a non-zero unique value in $(\tau_j^{\star},\bar{\tau}_j)$ at which we have $g(\alpha_j\eta_j)|\re{\alpha_j}|=|\re{\alpha_1}|$. Moreover, for $j\in\mathcal{I}_{\text{in}}$, we have $g(\alpha_j\eta_j)|\re{\alpha_j}| <|\re{\alpha_1}|$ for $\tau\in(\eta_j,\bar{\tau}_j)$ and $g(\alpha_j\eta_j)|\re{\alpha_j}| >|\re{\alpha_1}|$ for $\tau\in(0,\eta_j)$. From the aforementioned observations, the validity of the statements (a) and (b) follows from the continuity of $\rho_\tau$ in $\tau\in[0,\bar{\tau}]$, its definition~\eqref{eq::rho-tau-g} and also noting that the minimum of a set of strictly decreasing functions is also strictly decreasing.

Proof of statement (c): From the statement (b) we can conclude that $0<\tau^\star<\tilde{\tau}$. Given the definition of $\rho_\tau$ in~\eqref{eq::rho-tau-g}, we already know that $\rho_{\tau^\star}=\min\{g(\alpha_i\tau^\star)\,|\re{\alpha_i}|\}_{i=1}^n\leq g(\alpha_i\tau^\star)|\re{\alpha_i}|\leq g(\alpha_i\tau_i^\star)|\re{\alpha_i}|$, $i\in\{1,\cdots,n\}$. Therefore, $\tau^\star\leq \max\{\tau^\star_i\}_{i=1}^n$.  If $\mathcal{I}_{\text{out}}\neq\{\}$, because of $\tau^\star_j=0$, $j\in\mathcal{I}_\text{out}$, then $\tau^\star\geq \min\{\tau_i\}_{i=1}^n$ is trivial. Now assume $\mathcal{I}_{\text{out}}=\{\}$. In this case if $\tau^\star$ is not equal to any of the $\tau^\star_i$, $i\in\{1,\cdots,n\}$, then in order for $\tau^\star$ to be a maximizer point we should have non-empty $\bar{\mathcal{I}}\subsetneq \{1,\cdots,n\}$ and $\hat{\mathcal{I}}\subsetneq \{1,\cdots,n\}$ such that $\frac{\text{d}g(\alpha_i\tau^{\star})}{\text{d}\tau}>0$ for $i\in\bar{\mathcal{I}}$ and $\frac{\text{d}g(\alpha_i\tau^{\star})}{\text{d}\tau}<0$ for $i\in\hat{\mathcal{I}}$. Consequently, since for $\tau\in(0,\tau^\star_i)$ we have $\frac{\text{d}g(\alpha_i\tau)}{\text{d}\tau}>0$ and for $\tau\in(\tau^\star_i,\bar{\tau}_i)$ we have $\frac{\text{d}g(\alpha_i\tau)}{\text{d}\tau}<0$, for $i\in\mathcal{I}_{\text{in}}=\{1,\cdots,n\}$, then we can conclude that 
$\tau^\star\geq \min\{\tau_i^\star\}_{i=1}^n$.
\end{proof}

The statement (c) of Theorem~\ref{thm::com_alpha_n_rate} provides only an estimate on the location of $\tau^\star$. However, by relying on the proof argument of this statement we can narrow down the search for $\tau^\star$ to a set of discrete points as explained in the remark below. 
\begin{rem}[Candidate points for $\tau^\star$, when $\mathcal{I}_1\subset\mathcal{I}_{\text{in}}$]
Consider system~\eqref{eq::DDE-sys} when $\{\alpha_i\}_{i=1}^n\subset\complex_{-}^l$ and $\mathcal{I}_1\subset\mathcal{I}_{\text{in}}$. 
From the proof argument of the statement (c) of Theorem~\ref{thm::com_alpha_n_rate}, it follows that $\tau^\star$ is either a point in ${\mathcal{T}}^\star=\{\tau_i^\star\}_{i\in\mathcal{J}}$ where $\mathcal{J}=\{i\in\{1,\cdots,n\}|\min\{\tau^\star_k\}_{k=1}^n\leq\tau_i^\star<\tilde{\tau}\}$ or an intersection point of a $\rho_\tau^i$ and a $\rho_\tau^j$ in $\tau\in\big([\min\{\tau^\star_k\}_{k=1}^n,\max\{\tau^\star_k\}_{k=1}^n]\cap(0,\tilde{\tau})\big)$ 
where $
\frac{\text{d}g(\alpha_i\tau)}{\text{d}\tau}>0$ and $\frac{\text{d}g(\alpha_j\tau)}{\text{d}\tau}<0$. Based on this observation, we propose the following procedure to identify the candidate points for $\tau^\star$. Let $\mathcal{J}^{\text{r}}=\{i\in\mathcal{I}_{\text{in}}|\tau_i^\star\geq\tilde{\tau}\}$, and  $\mathcal{J}^{\text{l}}_j=\{i\in\mathcal{J}|\tau_i^\star\geq{\tau}_j^\star\}$ for any $j\in\mathcal{J}$. We note that for $i\in\mathcal{J}^{\text{r}}$ we have $\frac{\text{d}g(\alpha_i\tau)}{\text{d}\tau}>0$ for any $\tau\in(0,\tilde{\tau})$, and for any $j\in\mathcal{J}$ we have
$\frac{\text{d}g(\alpha_k\tau)}{\text{d}\tau}<0$ for any $\tau\in(\tau^\star_j,\tilde{\tau})$ and also $\frac{\text{d}g(\alpha_k\tau)}{\text{d}\tau}>0$ for any $\tau\in(\tau^\star_j,\tau^\star_k)\subset(\tau^\star_j,\tilde{\tau})$, $k\in\mathcal{J}^{\text{l}}_j$. Now for any $j\in\mathcal{J}$ let $\mathcal{T}_j$ be the set of intersection points of $\rho^j_\tau$ with $\rho^k_\tau$, $k\in (\mathcal{J}^{\text{l}}_j\cup \mathcal{J}^r)$ for $\tau\in(\tau^\star_j,\tilde{\tau})$ (here note that the possible intersection between $\rho_\tau^j$ and $\rho_\tau^k$ is in fact located at $(\tau^\star_j,\tau^\star_k)\subset(\tau^\star_j,\tilde{\tau})$ for $k\in\mathcal{I}^{\text{l}}_j$).
 Then, following the proof argument of the statement (c) of of Theorem~\ref{thm::com_alpha_n_rate}, we have $\tau^\star\in ((\underset{j\in\mathcal{J}}{\cup}\mathcal{T}_j)\cup {\mathcal{T}}^\star)$.
\end{rem}
Our next result shows that if $\mathcal{I}_1\not\subset\mathcal{I}_{\text{in}}$, not only $\rho_\tau<\rho_0$ but also $\rho_\tau$ is a strictly decreasing function of $\tau\in(0,\bar{\tau})$.

\begin{lem}[Rate of convergence when $\mathcal{I}_1\subset\mathcal{I}_{\text{out}}$]\label{lem::eigen_out}
Consider  system~\eqref{eq::DDE-sys} when $\{\alpha_i\}_{i=1}^n\subset\complex_{-}^l$. Recall the admissible delay bound $\bar{\tau}$ given by~\eqref{eq::delay_bound}. 
If $\mathcal{I}_1\subset\mathcal{I}_{\text{out}}$, then $\rho_\tau$ decreases strictly from $\rho_0$ to $0$ for $\tau\in[0,\bar{\tau}]$.
\end{lem}
\begin{proof}
Recall that $\bar{\tau}=\min\{\bar{\tau}_i\}_{i=1}^n$. Also recall the defintion of $\eta_i$ from the statement (a) of Theorem~\ref{lem::g_complexAlpha}. Next, note that due to the statement (a) of Theorem~\ref{lem::g_complexAlpha} for every $i\in\mathcal{I}_{\text{out}}$ we know that $\rho_\tau^i=g(\alpha_i\tau)|\re{\alpha_i}|$ is strictly decreasing for $\tau\in[0,\bar{\tau})$. We note that  because $\mathcal{I}_1\subset\mathcal{I}_{\text{out}}$, the previous statement holds for $i\in\mathcal{I}_1$. It also means that $\rho_\tau^i=g(\alpha_i\tau)|\re{\alpha_i}|<|\re{\alpha_1}|$ for $\tau\in[0,\bar{\tau})$. For $i\in\mathcal{I}_{\text{in}}$, from the proof argument of the statement (a) and (b) of Theorem~\ref{thm::com_alpha_n_rate} we know that 
$\rho_\tau^i=g(\alpha_i\tau)|\re{\alpha_i}|$ is decreasing for 
$\tau\in({\eta}_i,\bar{\tau}_i]$ and also that 
$\rho_\tau^i=g(\alpha_i\tau)|\re{\alpha_i}|\geq|\re{\alpha_1}$ for $\tau\in(0,{\eta}_i]$, which means that $\rho_\tau^i=g(\alpha_i\tau)|\re{\alpha_i}|\geq\rho_\tau^j$, $j\in\mathcal{I}_1$ for $\tau\in(0,{\eta}_i]$. Let $\mathcal{P}=\{\eta\in\real_{>0}|\eta=\eta_i~\text{if}~\eta_i<\bar{\tau},i\in\mathcal{I}_{\text{in}}\}$. If $\mathcal{P}=\{\}$, then we have $\rho_\tau=\min\{\rho_\tau^i\}_{i=1}^n=\min\{\rho_\tau^j\}_{j\in\mathcal{I}_{\text{out}}}$ for any $\tau\in[0,\bar{\tau})$. Therefore, since at delay interval $[0,\bar{\tau})$, $\rho_\tau$ is the minimum of strictly decreasing functions, it is also strictly decreasing. When $\mathcal{P}\neq\{\}$, let $\mathcal{P}=\{p_1,\cdots,p_{|\mathcal{P}|}\}$, where $p_m<p_n$ if $m<n$. Also, let $p_0=0$, $p_{|\mathcal{P}|+1}=\bar{\tau}$. Then, in light of the earlier observations, at each delay interval $[p_i,p_{i+1})$, $i\in\{0,\cdots,|\mathcal{P}|\}$, we have  $\rho_\tau=\min\{\rho_\tau^i\}_{i=1}^n=\min\{\rho_\tau^j\}_{j\in\mathcal{K}_i}$, where $\mathcal{K}_i=\mathcal{I}_{\text{out}}\cup\{k\in\mathcal{I}_{\text{in}}|\eta_k\leq p_i\}$. Since at each interval $[p_i,p_{i+1})$, $i\in\{0,\cdots,|\mathcal{P}|\}$, each $\rho_\tau^j$, $j\in\mathcal{K}_i$ is strictly decreasing, therefore $\rho_\tau$ is also strictly decreasing in delay interval $[p_i,p_{i+1})$. The proof then follows from $\cup_{i=1}^{|\mathcal{P}|}[p_i,p_{i+1})=[0,\bar{\tau})$.
\end{proof}


When all the eigenvalue of $\vect{A}$ are  negative reals numbers,
the directional derivative of $g(x)$ with respect to $\tau$ along $x=\alpha_i\tau$, for all $i\in\{1,\cdots,n\}$, follows a same pattern. This fact enables us to derive a simpler expression to compute $\tilde{\tau}$ and $\tau^\star$. We also show that $\rho_\tau$ depends only on $\alpha_1$ and $\alpha_n$.

\begin{lem}[When $\{\alpha_i\}_{i=1}^n\subset\real_{<0}$, $\rho_\tau$ of system~\eqref{eq::DDE-sys} depends only on $\alpha_1$ and $\alpha_n$]\label{lem::rho_matrix}
Consider system~\eqref{eq::DDE-sys} when $\{\alpha_i\}_{i=1}^n\subset\real_{<0}$. Recall the admissible delay bound $\bar{\tau}=\frac{\pi}{2|\alpha_n|}$ of this system from Lemma~\ref{thm::admis_tau}. Then, 
\begin{itemize}
            \item[(a)] $\rho_{\tau}=g(\alpha_1\tau)\,|\alpha_1|=-\frac{1}{\tau}\re{W_0(\alpha_1\tau)}$ for any  
          {\rm{$\tau\in(0,\tau^\star_n)\subset(0,\bar{\tau})$}},
     \item[(b)]  $\rho_{\tau}=\min\{g(\alpha_1\tau)\,|\alpha_1|,g(\alpha_n\tau)|\alpha_n|\}$ for any
     {\rm{$\tau\in([\tau^\star_n,\tau^\star_1]\cap(0,\bar{\tau}))$}},
    \item[(c)] If $\tau_1^\star<\bar{\tau}$, then $\rho_{\tau}=g(\alpha_n\tau)\,|\alpha_n|=-\frac{1}{\tau}\re{W_0(\alpha_n\tau)}$ for any
    {\rm{$\tau\in(\tau^\star_1,\bar{\tau})\subset(0,\bar{\tau})$}},
    \end{itemize}
    where {\rm{$\tau^\star_1=\frac{1}{|\alpha_1|\ee}$}} and {\rm{$\tau^\star_n=\frac{1}{|\alpha_n|\ee}$}}.
\end{lem}
\begin{proof}
Let $\mathcal{R}_r=[0,\frac{1}{\ee}]$ and $\mathcal{R}_l=[\frac{1}{\ee},\frac{\pi}{2})$. Then, note that from the statement (a) of Lemma~\ref{lem::alpha-real-positiverate}, we have that $g(-\xi)$ increases strictly from $1$ to $\ee$ for $\xi\in\mathcal{R}_r$. Therefore, 
     \begin{align}\label{eq::normalized_Rigion}
        & g(-\xi_1)<g(-\xi_2), \quad\text{if}~~(\xi_1<\xi_2~\text{and}~\xi_1,\xi_2\in\mathcal{R}_r).
     \end{align}
Next, let $\mu=g(-\xi)\,\xi$. We note that  $\frac{\text{d}\mu}{\text{d}\xi}=\frac{\text{d}g(-\xi)}{\text{d}\xi}\xi+g(-\xi)$.  For $\xi\in(\frac{1}{\ee},\frac{\pi}{2})$, from~\eqref{eq::drev_g_at_tau_star} and the manipulations leading to it, we can write   $\frac{\text{d}\mu}{\text{d}\xi}=\frac{1}{\xi}\,\frac{\mathsf{u}^2(-\mathsf{u}\frac{\cos(\mathsf{u})}{\sin(\mathsf{u})}+\cos(2\,\mathsf{u}))}{(-\mathsf{u}\,\cos(\mathsf{u})+\sin(\mathsf{u}))^2+\mathsf{u}^2\,\sin^2(\mathsf{u})}+\frac{1}{\xi}\frac{\mathsf{u}\cos(\mathsf{u})}{\sin(\mathsf{u})}
   =\frac{1}{2\xi}\frac{-\mathsf{u}(2\mathsf{u}-\sin(2\mathsf{u}))}{(-\mathsf{u}\,\cos(\mathsf{u})+\sin(\mathsf{u}))^2+\mathsf{u}^2\,\sin^2(\mathsf{u})}$, 
where $\mathsf{u}=\im{W_0(-\xi)}$. Consequently, for  $\xi\in(\frac{1}{\ee},\frac{\pi}{2})$, since  $\mathsf{u}\in(0,\frac{\pi}{2})$ and therefore  $\sin(2\mathsf{u})\leq2\mathsf{u}$, we get
$\frac{\text{d}\mu}{\text{d}\xi}<0$. This conclusion along with $\mu$ being a continuous function in $\xi\in\mathcal{R}_l$, confirms that
    \begin{align}\label{eq::normalized_Rigion2}
        & g(-\xi_1)\xi_1>g(-\xi_2)\xi_2, \quad\text{if}~(\xi_1<\xi_2 ~\text{and}~\xi_1,\xi_2\in\mathcal{R}_l).
  \end{align}
Proof of statement (a):
Since according to the statement (a) of Lemma~\ref{lem::alpha-real-positiverate},  we have $|\alpha_i|\tau^\star_i=\frac{1}{\ee}$ for $i\in\{1,\cdots,n\}$, then $0<{\tau}_{n}^\star\leq{\tau}_{n-1}^\star\leq \cdots\leq{\tau}^{\star}_1$.
This fact along with $|\alpha_1|\tau\leq\cdots\leq |\alpha_n|\tau$ lead us to conclude from~\eqref{eq::normalized_Rigion} that
\begin{align}\label{eq:sta,c}
    g(\alpha_1\tau)\!\leq\! g(\alpha_2\tau)\!\leq \cdots\!\leq\! g(\alpha_n\tau),~~ \tau\!\in\!(0,\tau_n^\star].
\end{align}
Here we used the fact that for $\tau\in(0,\tau_n^\star]$, we have $|\alpha_i|\tau\in\mathcal{R}_r$ and for $\tau\in[\tau_1^\star,\bar{\tau}]$, we have $|\alpha_i|\tau\in\mathcal{R}_l$, $i\in\{1,\cdots,n\}$. Now, given~\eqref{eq:sta,c}, we have $\rho_\tau=\min\{g(\alpha_i\tau)|\alpha_i|\}_{i=1}^n=g(\alpha_i\tau)|\alpha_1|$, which completes the proof of  the statement (a).

Proof of statement~(b): For a $\tau\in[\tau^\star_n,\tau^\star_1]$, let $\mathcal{I}_r=\{i\in\{1,\cdots,n\}||\alpha_j|\tau\in\mathcal{R}_r\}$ and $\mathcal{I}_l=\{i\in\{1,\cdots,n\}||\alpha_j|\tau\in\mathcal{R}_l\}$. Surely, for any $\tau\in[\tau^\star_n,\tau^\star_1]$, $1\in\mathcal{R}_r$ and $n\in\mathcal{R}_l$, however for $j\in\{2,\cdots,n-1\}$, $|\alpha_j|\tau$, depending on its value, can be in either in $\mathcal{R}_l$ or $\mathcal{R}_r$. Following the same proof argument of the statement (a), then for any $\tau\in[\tau^\star_n,\tau^\star_1]$, we have $\min\{g(\alpha_i\tau)|\alpha_i|\}_{i\in\mathcal{I}_{l}}=g(\alpha_1\tau)|\alpha_1|$. To complete the proof of the statement (b), by taking into account the definition of $\rho_\tau$ in~\eqref{eq::rho-tau-g}, next we show that for any $\tau\in[\tau^\star_n,\tau^\star_1]$, we have $\min\{g(\alpha_i\tau)|\alpha_i|\}_{i\in\mathcal{I}_{r}}=g(\alpha_n\tau)|\alpha_n|$. For this, we note that since $0<{\tau}_{n}^\star\leq{\tau}_{n-1}^\star\leq \cdots\leq{\tau}^{\star}_1$ and $|\alpha_1|\tau\leq\cdots\leq |\alpha_n|\tau$, we can conclude from~\eqref{eq::normalized_Rigion2} that for any $\{j,k\}\subset \mathcal{I}_{r}$ such that $j>k$ we have
\begin{align}\label{eq:stb}
    g(\alpha_j\tau)|\alpha_j|\tau\!\leq\! g(\alpha_k\tau)|\alpha_k|\tau,~~ \tau\!\in\![\tau_1^\star,\tau_n^\star].
\end{align}
Therefore, we can write $\min\{g(\alpha_i\tau)|\alpha_i|\}_{i\in\mathcal{I}_{r}}=g(\alpha_n\tau)|\alpha_n|$.

 Proof of the statement (c) follows directly from the proof of the statement (b), by noting that for $\tau\in[\frac{1}{\ee|\alpha_1|},\frac{\pi}{2\,|\alpha_n|})\subset(0,\bar{\tau})$, we have $\alpha_i\tau\in\mathcal{R}_l$ for all $i\until{n}$.
\end{proof}

Next we show that when $\{\alpha_i\}_{i=1}^n\subset\real_{<0}$, $\tilde{\tau}$, which according to the statement~(a) of Theorem~\ref{thm::com_alpha_n_rate} is equal to $\min\{\eta_i\}_{i=1}^n$, is in fact given by $\min\{\eta_1,\eta_n\}$. The result below also gives a close form solution for $\tau^\star$ and its corresponding $\rho_\tau^\star$.

\begin{thm}[Rate of convergence of~\eqref{eq::DDE-sys} with and without delay when $\{\alpha_i\}_{i=1}^n\subset\real_{<0}$]\label{thm::main_thm}
Consider  system~\eqref{eq::DDE-sys} when $\{\alpha_i\}_{i=1}^n\subset\real_{<0}$. Recall the admissible delay bound $\bar{\tau}=\frac{\pi}{2|\alpha_n|}$ of this system from Lemma~\ref{thm::admis_tau}. Then, 
\begin{itemize}
    \item[(a)] 
    $\tilde{\tau}=\min\{\tilde{\tau}_1,\eta_n\}$,  where $\eta_n$ is defined in the statement (a) of Theorem~\ref{thm::com_alpha_n_rate}. Moreover, $\tau^\star_n<\tilde{\tau}<\bar{\tau}$.
 \item[(b)] 
 the maximum rate of convergence of 
 \begin{align}\label{eq::rate_real_max}\rho_\tau^\star=\ee^{\frac{\arccos(\frac{\alpha_1}{\alpha_n})}{\sqrt{(\frac{\alpha_n}{\alpha_1})^2-1}}}\,|\alpha_1|,\end{align}
  is attained at \begin{align}\label{eq::opt_time_delay}
\tau^\star=\frac{\arccos(\frac{\alpha_1}{\alpha_n})}{|\alpha_1|\sqrt{\frac{\alpha_n}{\alpha_1}^2-1}}\ee^{-\frac{\arccos(\frac{\alpha_1}{\alpha_n})}{\sqrt{(\frac{\alpha_n}{\alpha_1})^2-1}}},
     \end{align}
     \end{itemize}
        where {\rm{$\tau^\star_1=\frac{1}{\ee|\alpha_1|}$}} and {\rm{$\tau^\star_n=\frac{1}{\ee|\alpha_n|}$}}. Moreover, $\tau^\star\in([\tau^{\star}_n,\tau^{\star}_1]\cap[\tau_n^{\star},\bar{\tau}))$.
\end{thm}
\begin{proof}
Lemma~\ref{lem::rho_matrix} showed that
\begin{align}\label{eq::rate_alpha1N}
    \rho_\tau=\min\{g(\alpha_1\tau)\,|\alpha_1|,g(\alpha_n\tau)\,|\alpha_n|\},\quad \tau\in(0,\bar{\tau}). 
\end{align}
We will use this fact to prove our statements.

To prove the statement (a), first note that by definition we have $g(\alpha_1\tilde{\tau}_1)=1$ and $g(\alpha_1\eta_n)=\frac{\alpha_1}{\alpha_n}$. Therefore, we have $g(\alpha_1\tilde{\tau}_1)|\alpha_1|=g(\alpha_1\eta_n)|\alpha_n|=|\alpha_1|$. Next, note that in the proof of the statement (a) of Theorem~\ref{thm::com_alpha_n_rate} we have already shown that $\eta_n$ satisfies ${\tau}^{\star}_n<\eta_n<\bar{\tau}_n=\bar{\tau}$, which means $\tau\in(\eta_n,\bar{\tau})\subset[\tau_n^{\star},\bar{\tau})$.
Therefore, by virtue of the statement (a) of Lemma~\ref{lem::alpha-real-positiverate}, we have $g(\alpha_n\tau)|\alpha_n|<g(\alpha_n\eta_n)|\alpha_n|=|\alpha_1|$ for $\tau\in(\eta_n,\bar{\tau})\subset[\tau_n^{\star},\bar{\tau})$. Subsequently, from~\eqref{eq::rate_alpha1N}, we conclude that $\rho_{\tau}<|\alpha_1|$ for $\tau\in(\eta_n,\bar{\tau})$. Next, note that for any $\tau\in(\tilde{\tau}_1,\bar{\tau})$ by virtue of the statement (b) of Lemma~\ref{lem::alpha-real-positiverate},  we have $g(\alpha_1\tau)<1$ and consequently, $g(\alpha_1\tau)|\alpha_1|<|\alpha_1|$. Therefore,  because of~\eqref{eq::rate_alpha1N}, we have also the guarantee that $\rho_\tau<\rho_{0}=|\alpha_1|$ for any $\tau\in(\tilde{\tau}_1,\bar{\tau})$. As a result, we can conclude from~\eqref{eq::rate_alpha1N} that $\rho<\rho_0$ for $\tau\in(\min\{\tilde{\tau}_1,\eta_n\},\bar{\tau})$. Using Lemma~\ref{lem::alpha-real-positiverate}, we have the guarantees that $g(\alpha_n\tau)>1$ for $\tau\in(0,\tau^{\star}_n]$ and $g(\alpha_n\tau)$ is strictly decreasing for $\tau\in[\tau^{\star}_n,\bar{\tau}_n)$. Therefore,  $g(\alpha_n\tau)>\frac{\alpha_1}{\alpha_n}$ for $\tau\in(0,\eta_n)$ (note here that $\tau^\star_n<\tilde{\tau}_n
\leq \eta_n<\bar{\tau}$).
From Lemma~\ref{lem::alpha-real-positiverate}, we also know that $g(\alpha_1\tau)>1$
for $\tau\in(0,\tilde{\tau}_1)$. As a result, we can conclude from~\eqref{eq::rate_alpha1N} that $\rho_\tau>|\alpha_1|$ for $\tau\in(0,\min\{\tilde{\tau}_1,\eta_n\})$. This completes the proof of statement (a).

Using~\eqref{eq::rate_alpha1N}, we proceed to prove our statement (b) as follows.
First we consider the case that $\alpha_1=\alpha_n$ in which the rate of convergence at $\tau\in(0,\bar{\tau})$ is $\rho_\tau=g(\alpha_1\tau)\,|\alpha_1|$. Here, by virtue of Lemma~\ref{lem::alpha-real-positiverate} one can see that the maximum rate of $\rho_\tau^\star=\ee|\alpha_1|$ is attained at  $\tau^\star=\frac{1}{\ee|\alpha_1|}$ (here note that $\tau_1^\star=\tau_n^\star=\tau^\star$). Then, for the case of $\alpha_1=\alpha_n$ the proof of the statement (b) follows from $\lim_{\frac{\alpha_1}{\alpha_n}\to1}\big({\frac{\arccos(\frac{\alpha_1}{\alpha_n})}{\sqrt{(\frac{\alpha_n}{\alpha_1})^2-1}}}\big)=1$. 

Next, we consider the case where $|\alpha_n|>|\alpha_1|$. 
From the proof of the statement (a), we know that $\tau^\star$ should satisfy
$\tau^\star\in(0,\tilde{\tau})$. Also, recall that $\tau_n^\star<\tilde{\tau}$. Since $\tau_n^\star<\tau_1^\star$, by virtue of the statement (a) of Lemma~\ref{lem::alpha-real-positiverate} we know that both $g(\alpha_1\tau)$ and $g(\alpha_n\tau)$ are strictly increasing for $\tau\in[0,\tau^\star_n)$. Therefore, from~\eqref{eq::rate_alpha1N} we can conclude that $\rho_\tau$ is also strictly increasing in $\tau\in[0,\tau^\star_n)$. Hence, $\tau^\star\geq\tau^\star_n$.
If $\tau^\star_1<\bar{\tau}$, then by virtue of the statement (a) of Lemma~\ref{lem::alpha-real-positiverate} we know that both $g(\alpha_1\tau)$ and $g(\alpha_n\tau)$ are strictly decreasing for $\tau\in[0,\tau^\star_n)$. Therefore, from~\eqref{eq::rate_alpha1N} we can conclude that $\rho_\tau$ is also strictly decreasing in $\tau\in(\tau^\star_1,\bar{\tau})$. Hence, $\tau^\star\leq\tau^\star_1$. If $\tau^\star_1>\bar{\tau}$, we know that $\rho_\tau<0$ and the system is unstable for $\tau>\tau_1^\star$. 

So far we have shown that $\tau^\star\in([\tau^{\star}_n,\tau^{\star}_1]\cap[\tau_n^{\star},\bar{\tau}))$. From the discussions for far we also know that $g(\alpha_1\tau)|\alpha_1|$ is strictly increasing for $\tau\in[\tau_n^\star,\tau_n^\star]$, and $g(\alpha_n\tau)|\alpha_n|$ is strictly decreasing in $[\tau_n^\star,\bar{\tau})$. Therefore, from~\eqref{eq::rate_alpha1N} we conclude that at $\tau^\star$ we have
    $g(\alpha_1\tau^\star)|\alpha_1|=g(\alpha_n\tau^\star)|\alpha_n|$,
    or equivalently
    (recall~\eqref{eq::delay_rate_gain}) when 
    \begin{align}\label{eq::real_max_con}
    \re{W_0(\alpha_1\,\tau^\star)}=\re{W_0(\alpha_n\,\tau^\star)}.
    \end{align}
    Because for $\tau\in([\tau^{\star}_n,\tau^{\star}_1]\cap[\tau_n^{\star},\bar{\tau}))$, we have  $\alpha_1\,\tau^\star\in[\alpha_1\min\{\frac{1}{|\alpha_1|\ee},\frac{\pi}{2|\alpha_n|}\},-\frac{\alpha_1}{\alpha_n\ee}]\subset[-\frac{1}{\ee},0)$ , then $W_0(\xi_1)$ is a negative real number, i.e., $\re{W_0(\alpha_1\,\tau^\star)}=W_0(\alpha_1\,\tau^\star)=\mathsf{w}_1\in\real_{<0}$. Subsequently, from~\eqref{eq::real_max_con} we have
    $W_0(\alpha_n\tau^\star)=\mathsf{w}_1+{\ii}\,\mathsf{u}$ for some $\mathsf{u}\in(-\pi,\pi)$. 
    Therefore, we can write
     \begin{align*}
     &W_0(\alpha_1\,\tau^\star)=\mathsf{w}_1\qquad~~~~\rightarrow~ \mathsf{w}_1\,\ee^{\mathsf{w}_1}=\alpha_1\,\tau^\star,\\
     &W_0(\alpha_n\tau^\star)=\mathsf{w}_1+{\ii}\,\mathsf{u}~\,\rightarrow~ (\mathsf{w}_1+{\ii}\,\mathsf{u})\,\ee^{\mathsf{w}_1+{\ii}\,\mathsf{u}}=\alpha_n\tau^\star,
     \end{align*}        
     which by eliminating $\tau^\star$ gives 
       \begin{align*}
    \mathsf{w}_1\cos(\mathsf{u})-\mathsf{u}\sin(\mathsf{u})=\frac{\alpha_n}{\alpha_1}\mathsf{w}_1,\text{~and~}\mathsf{u}\cos(\mathsf{u})+\mathsf{w}_1\sin(\mathsf{u})=0.
     \end{align*}
     Then, using $\cos(\mathsf{u})^2+\sin(\mathsf{u})^2=1$, we obtain $\cos(\mathsf{u})=\frac{\alpha_1}{\alpha_n}$ and $\sin(\mathsf{u})=\frac{(\alpha_1^2-\alpha_n^2)\,\mathsf{w}_1}{\alpha_n\alpha_1\,\mathsf{u}}$. Subsequently, because $\frac{\alpha_n}{\alpha_1}\in(1,\infty)$, we obtain $\mathsf{u}=\arccos(\frac{\alpha_1}{\alpha_n})\in(0,\pi)\subset(-\pi,\pi)$ and $\sin(\mathsf{u})=\frac{1}{|\alpha_n|}\sqrt{\alpha_n^2-\alpha_1^2}$ and $\mathsf{w}_1=-\frac{\arccos(\frac{\alpha_1}{\alpha_n})}{\sqrt{(\frac{\alpha_n}{\alpha_1})^2-1}}$. Then, by virtue of $\mathsf{w}_1\,\ee^{\mathsf{w}_1}=\alpha_1\,\tau^\star$,~\eqref{eq::opt_time_delay} is confirmed. Finally,~\eqref{eq::rate_real_max} is confirmed by  $\rho^\star_\tau=-g(\alpha_1\tau^\star)\,\alpha_1=-\frac{1}{\tau^\star}\re{W_0(\alpha_1\tau^\star)}=-\frac{\mathsf{w}_1}{\tau^\star}$.
\end{proof}

\begin{rem}[Ultimate bound on the maximum possible increase in the rate of convergence of system~\eqref{eq::DDE-sys} in the presence of time delay]
We note that the suprimum value of ${\arccos(\gamma)}/{\sqrt{\frac{1}{\gamma^2}-1}}$
for $\gamma\in(0,1)$ is $1$. Therefore, $\rho_\tau^\star$ in~\eqref{eq::rate_real_max} is always less than or equal to $\ee|\alpha_1|=\ee\rho_0$, regardless of the value of $\alpha_1\in\real_{<0}$ and $\alpha_n\in\real_{<0}$. The same result, when $\{\alpha_i\}_{i=1}^n\subset\real_{<0}$, was established in~\cite{WQ-RS:13} using a different approach. Inspecting the contour plots of $g(\mathsf{x}+\mathsf{y}{\ii})$
 in Fig.~\ref{Fig::g_countor} reveals that the maximum attainable value for $g(x)$ for any $x\in\complex_{-}^l$ is $\ee$. Therefore, given the alternative definition of $\rho_\tau$ in~\eqref{eq::rho-tau-g}, we conjecture that in fact the maximum rate of convergence due to delay when $\{\alpha_i\}_{i=1}^n\subset\complex_{-}^l$ is also $\ee \rho_0$. 
\end{rem}

\vspace{-0.1in}
\begin{rem}[System design for faster convergence]
Equation~\eqref{eq::rate_real_max} indicates that for $\frac{\alpha_1}{\alpha_n}\rightarrow1$ (compact eigenvalue spectrum), higher convergence rate can be achieved due to delay. 
This fact can be used in system design to make the best of accelerated convergence due to delay. 
For example, in case of the average consensus algorithm in connected networks (see Section~\ref{sec::num_ex}), $-\alpha_1$ and $-\alpha_n$ correspond to the smallest ($\lambda_2$) and the largest ($\lambda_N$) non-zero eigenvalues of the graph Laplacian. There exists known relations between the graph topology and the magnitude of these eigenvalues~\cite{ros-jaf-rmm:07,SSK-BVS-JC-RAF-KML-SM:19}. Graph topological design can then be used to make $\frac{\lambda_2}{\lambda_N}$ closer to $1$. Or in case of a state feedback control design for $\dvect{x}=\vect{A}\vect{x}(t-\tau)+\vect{B}\vect{u}$, a delayed feedback  $\vect{u}=\vect{K}\vect{x}(t-\tau)$ can be used to place the eigenvalues of the closed-loop system matrix in a compact and negative real spectrum. 
\end{rem}

\vspace{-0.1in}
\section{Demonstrative example
}\label{sec::num_ex}\vspace{-0.05in}
We demonstrate our results by studying the effect of delay on the static average consensus algorithm~\cite{ros-jaf-rmm:07} for a group of $N$ networked agents interacting over a strongly connected and weight-balanced directed graph (or simply digraph), similar to the one shown in Fig.~\ref{fig:network} (we follow~\cite{FB-JC-SM:09} for graph related terminologies and definitions). The set of all agents that can send information to agent $i$ are called its out-neighbors. 
  A digraph is strongly connected if there is a directed path from every agent $i$ to every agent $j$ in the graph. Let $\vectsf{W}=[\mathsf{w}_{ij}]\in\real^{N\times N}$ be the adjacency matrix of a given digraph, defined according to $\mathsf{w}_{ii}=0$, $\mathsf{w}_{ij}>0$ if agent $j$ can send information to agent $i$, and $\mathsf{w}_{ij}=0$ otherwise. A digraph of $N$ agents is weight-balanced if and only if $\sum_{j=1}^N\mathsf{w}_{ij}=\sum_{j=1}^N\mathsf{w}_{ji}$ for any $i\in\{1,\cdots,N\}$. Let every agent in this network have a local reference value $\mathsf{r}^i\in\real$, $i\in\{1,\cdots,N\}$. The static average consensus problem consists of designing a distributed algorithm that enables each agent to obtain  $\frac{1}{N}\sum_{j=1}^N \mathsf{r}^j$ by using the information it only receives from its out-neighbors.
As shown in~\cite{ros-jaf-rmm:07}, for strongly connected and weight-balanced digraphs, the  Laplacian dynamics 
\begin{align*}
    \dot{x}^i(t)\!=\!\sum\nolimits_{j=1}^N\!\!\mathsf{w}_{ij}(x^j(t)-x^i(t)),~ x^i(0)=\mathsf{r}^i,~i\!\in\!\{1,\cdots,N\},
\end{align*}
is guaranteed to satisfy $x^i\to\frac{1}{N}\sum_{j=1}^N \mathsf{r}^j$, as $t\to\infty$.  Using $\vect{x}=[x^1,\cdots,x^N]^\top$, the compact form of the Laplacian dynamics in the presence of delay $\tau\in\real_{>0}$ is 
\begin{align}
   & \dvect{x}(t)=-\vect{L}\,\vect{x}(t-\tau),\label{eq:laplace}\\
   &x^i(t)=\phi^i(t)\in\real,~~ t\in [-\tau,0],~\phi^i(0)=\mathsf{r}^i,~~i\in\{1,\cdots,N\},\nonumber
\end{align}
where $\vect{L}\!=\!\text{diag}(\vectsf{W}\vect{1}_N)-\vectsf{W}$. 
For strongly connected and weight-balanced digraphs, we have $\rank(\vect{L})\!=\!N-1$, $\vect{L}\vect{1}_N\!=\!\vect{0}$, $\vect{1}^\top_N\vect{L}=\vect{0}$. Moreover, $\vect{L}$ has one simple zero eigenvalue $\lambda_1=0$ and the rest of its eigenvalues $\{\lambda_j\}_{j=2}^N$ has negative real parts~\cite{FB-JC-SM:09}.
Next, consider the change of variable $\vect{y}=\vect{T}^\top\vect{x}$ with $\vect{T}=\begin{bmatrix}\frac{1}{\sqrt{N}}\vect{1}_N&\vect{R}\end{bmatrix}$, where $\vect{R}$ is such that $\vect{T}^\top\vect{T}=\vect{T}\vect{T}^\top=\vect{I}_N$. Then the Laplacian dynamics can be represented in the following equivalent form
\begin{subequations}\label{eq::laclacian_equivalent}
\begin{align}
    \dot{y}_1(t)&=0,\quad\quad\quad {y}_1(0)=\frac{1}{\sqrt{N}}\sum\nolimits_{i=1}^N\mathsf{r}^i,\\
    \dvect{y}_{2:N}(t)&=\vect{A}\vect{y}_{2:N}(t-\tau).\label{eq::laclacian_equivalent_y2}
\end{align}
\end{subequations}
where $\vect{y}\!=\![y_1^\top~\vect{y}_{2:N}^\top]^\top$ and $\vect{A}\!=\!-(\vect{R}^\top\vect{L}\vect{R})$. The matrix $\vect{A}$ is Hurwitz with eigenvalues $\{\alpha_j\}_{j=1}^n\!=\!\{\lambda_i\}_{i=2}^N\subset\complex_{-}^l$, $n\!=\!N\!-\!1$. Since, $\lim_{t\to\infty}\vect{x}(t)\!=\!
(\frac{1}{N}\sum\nolimits_{i=1}^N\mathsf{r}^i)\vect{1}_N+\vect{R}\lim_{t\to\infty}\vect{y}_{2:N}(t)$,
the correctness and the convergence rate of the average consensus algorithm~\eqref{eq:laplace} are determined, respectively, by exponential stability and the convergence rate of~\eqref{eq::laclacian_equivalent_y2}. Since the time-delayed system~\eqref{eq::laclacian_equivalent_y2} is in the form of our system of interest~\eqref{eq::DDE-sys}, the effect of delay and how it can potentially be used to accelerate the rate of convergence of the algorithm~\eqref{eq:laplace} can be fully analyzed by the results described in Section~\ref{sec::main}. 
\begin{figure}[t]
  \unitlength=0.5in
  \centering
\includegraphics[trim=13pt 0 0 0,clip,scale=0.5]{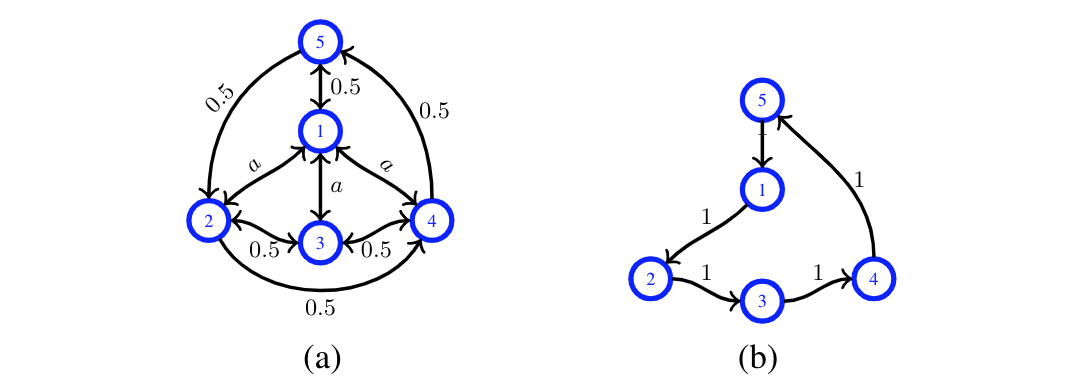}
\vspace{-0.1in}
       \caption{\small{Strongly connected and weight-balanced networks with their corresponding connection weights. An arrow from agent $i$ to agent $j$ means that agent $i$ can obtain information from agent $j$.}
       }
    \label{fig:network}\vspace{-0.15in}
\end{figure}

\begin{figure}[t]
    \centering
        \includegraphics[trim=15pt 0 0 0,clip,scale=0.55]{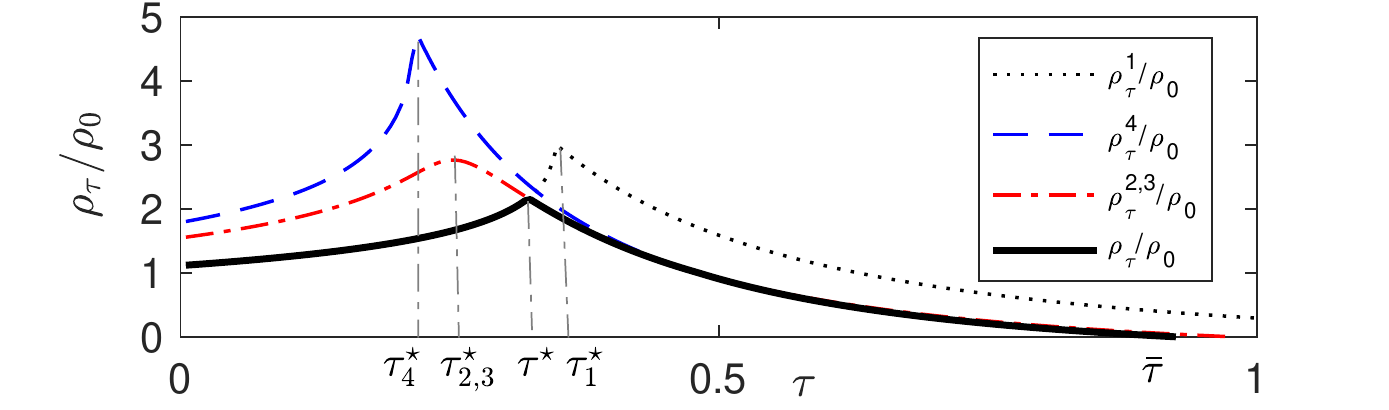}
   \caption{\small
   The Normalized rate of convergence versus time delay for different modes of  system~\eqref{eq::laclacian_equivalent} for case~(III). $\rho_{\tau}^1$, $\rho_{\tau}^2$, $\rho_{\tau}^3$ and $\rho_{\tau}^4$ are the rate of convergence corresponding to $\alpha_1=-1.05$,
    $\alpha_{2,3}=-1.47\pm0.18\ii$ and $\alpha_4=-1.70$, respectively.  Recall that according to~\eqref{eq::rho} we have $\rho_\tau/\rho_0=\min\{\rho_{\tau }^i/\rho_0\}_{i=1}^4$, which its normalized value is shown by the thick black curve.
   }\vspace{-0.1in}
    \label{fig:resp2}
\end{figure}
For numerical study, we use the digraphs in~Fig.~\ref{fig:network}.
We note that 
(I) $\text{eig}(\vect{A})\!=\!\{\alpha_j\}_{j=1}^4\!=\!\{-1.50,-1.50,-2.00,-2.50\}$ for the digraph in Fig.~\ref{fig:network}(a) when $a\!=\!0.50$
(II) $\text{eig}(\vect{A})\!=\!\{\alpha_j\}_{j=1}^4\!=\!\{-0.69+0.95\ii, -0.69\!-\! 0.95\ii, -1.80+ 0.58\ii,
   -1.80- 0.58\ii
   \}$ for digraph in Fig.~\ref{fig:network}(b),
   and (III) $\text{eig}(\vect{A})\!=\!\{\alpha_j\}_{j=1}^4\!=\!\{-1.05,-1.47\!+\!0.18\ii,-1.47\!-\!0.18\ii,-1.70\}$ for the digraph in Fig.~\ref{fig:network}(a) when $a\!=\!0.20$.

  Fig.~\ref{fig::cont_ex} shows $\rho_\tau/\rho_0$ 
  versus time delay for the cases~(I) and~(II).
   For the case~(I) we have $\{\alpha_i\}_{i=1}^4\subset\real_{<0}$ and also $\mathcal{I}_\text{1}=\mathcal{I}_\text{in}=\{1,2,3,4\}$. Hence, as predicted by  Theorem~\ref{thm::main_type_complex},
   there exists a $\tilde{\tau}\in\real_{>0}$ such that $\rho_\tau>\rho_0$ for $\tau\in(0,\tilde{\tau})\subset(0,\bar{\tau})$. 
   In this case,  $\bar{\tau}=\frac{\pi}{2|\alpha_4|}=0.63$ (marked as $\bar{\tau}_{I}$ on y-axis of Fig.~\ref{fig::cont_ex}). Moreover, $\tilde{\tau}$, following the statement (a) of Theorem~\ref{thm::main_thm}, we get $\tilde{\tau}=\min\{\tilde{\tau}_1=0.71,\eta_4=0.32\}= 0.32$, which is exactly the same value that one reads on Fig.~\ref{fig::cont_ex}, marked as $\tilde{\tau}_{I}$ on y-axis. Also, the maximum rate of convergence is attained at $\tau^\star=0.23$ (marked as ${\tau}^\star_{I}$ on y-axis of Fig.~\ref{fig::cont_ex}), which can be obtained from~\eqref{eq::opt_time_delay} in the statement (b) of Theorem~\ref{thm::main_thm}. At $\tau^\star$, the maximum attainable rate of convergence can be obtained from~\eqref{eq::rate_real_max} in the statement (b) of Theorem~\ref{thm::main_thm} to be $\rho_\tau=1.98\rho_0$, which matches the value one reads on~Fig.~\ref{fig::cont_ex}. 
   In the case~(II), we have $\{\alpha_i\}_{i=1}^4\subset\complex^l_-$
   ,   $\mathcal{I}_\text{in}=\{3,4\}$ and $\mathcal{I}_1=\mathcal{I}_\text{out}=\{1,2\}$ where $\mathcal{I}_1$, $\mathcal{I}_{\text{in}}$ and $\mathcal{I}_{\text{out}}$ are defined by~\eqref{eq::def_I}. Therefore, since $\mathcal{I}_1\subset\mathcal{I}_\text{out}$, as predicted by~Lemma~\ref{lem::eigen_out}, $\rho_\tau$  decreases strictly with delay delay until it reaches $0$ at $\bar{\tau}=0.51$, as shown in~Fig.~\ref{fig::cont_ex}. 
  
   For case (III), we have $\{\alpha_i\}_{i=1}^4\in\complex_{-}^l$ and $\mathcal{I}_1=\{1\}\subset\mathcal{I}_{\text{in}}=\{1,2,3,4\}$. Therefore, according to Theorem~\ref{thm::main_type_complex}, we expect 
   existence of $\tilde{\tau}\in\real_{>0}$ such that $\rho_\tau>\rho_0$ for $\tau\in(0,\tilde{\tau})\subset(0,\bar{\tau})$, which is in accordance with the trend one observes for $\rho_\tau$ in Fig.~\ref{fig:resp2}. Also, as seen in Fig.~\ref{fig:resp2}, we have $\rho_\tau=\min\{\rho_{\tau}^i\}_{i=1}^4=\min\{\rho_{\tau}^{1},\rho_{\tau}^{2,3}\}$ for any $\tau\in[0,\bar{\tau}]$. Here, $\bar{\tau}=0.92$, and  $\tilde{\tau}=0.46$, which as expected from the statement (a) of Theorem~\ref{thm::com_alpha_n_rate}, is the minimum of $\eta_1=0.59$, $\eta_2=0.46$,  $\eta_3=0.46$ and $\eta_4=0.47$. Moreover, as expected from Theorem~\ref{thm::com_alpha_n_rate}, the value of $\tau^\star$ satisfies $\tau^\star\in[\tau^\star_4,\tau^\star_1]\cap[0,\tilde{\tau})$ as shown in Fig.~\ref{fig:resp2} where $\tau^\star_1=\frac{1}{|\alpha_1|\ee}=0.35=\min\{\tau^{\star}_i\}_{i=1}^4$ and $\tau^\star_4=\frac{1}{|\alpha_4|\ee}=0.21=\max\{\tau^{\star}_i\}_{i=1}^4$. In this case the maximum attainable rate is $\rho_{\tau^\star}=1.92\rho_0$.

\section{Conclusion and future work}\label{sec::conclu}\vspace{-0.05in}
We examined the effect of a fixed time delay on the rate of convergence of a class of time-delayed LTI systems to address the following fundamental questions (a) what systems can experience increase in their rate of convergence due to delay (b) for what values of delay the rate of convergence is increased due to delay (c) what is the maximum achievable rate due to delay and its corresponding maximizing delay value. 
Our analysis relied on use of the Lambert W function to specify the rate of convergence of our time-delayed LTI system of interest. 
We validated our result through a numerical example on accelerating the agreement algorithm for a network of multi-agent systems. 
Our future work is focused on expanding our results to a wider class of time-delayed LTI systems, and also exploring the application of our theoretical results in design of fast-converging distributed algorithms for networked systems.

\bibliographystyle{ieeetr}%
\bibliography{bib/alias,bib/Reference}

\appendix
\renewcommand{\theequation}{A.\arabic{equation}}
\renewcommand{\thethm}{A.\arabic{thm}}
\renewcommand{\thelem}{A.\arabic{lem}}
\renewcommand{\thedefn}{A.\arabic{defn}}

\setcounter{equation}{0}

\begin{lem}[$\rho_\tau$ is a continuous function of $\tau$]\label{lem::p_tau_continuous}
The rate of convergence $\rho_\tau$ of system~\eqref{eq::DDE-sys} given by~\eqref{eq::rho} is a continuous function of $\tau\in\real_{\geq0}$.
\end{lem}
\begin{proof}
For any given $\alpha\in\complex$, $\re{W_0(\alpha\tau)}$ is a continuous function of $\tau\in\real_{\geq 0}$. Moreover, for any $\alpha\in\complex_{-}^l$, by virtue of~\eqref{eq::limWz-z} we have  $\lim_{\tau\to0^-}\frac{\re{W_0(\alpha\tau)}}{\tau}=1$. Therefore, for every $\alpha_i$, $i\until{n}$, $\frac{\re{W_0(\alpha_i\tau)}}{\tau}$ is continuous over $\tau\in\real_{\geq0}$.
Then, the proof is follows from the fact that the maximum/minimum of continuous functions is a continuous function (c.f.~\cite[Problem~1.2.13]{WJK-MTN:01}).   
\end{proof}

The rest of this appendix contains the proof of the lemmas of Section~\ref{sec::Prob_formu}.


\begin{proof}[Proof of Lemma~\ref{lem::g0-g1}]
Given the definition of $g(x)$ in~\eqref{eq::gain_def}, the proof of~\eqref{eq::lim-g-0} follows directly from~\eqref{eq::limWz-z}.

To validate~\eqref{eq::tau_bar}, we proceed as follows.
Note that $g(\alpha\tau)=0$ requires  $\re{W_{0}(\alpha\tau)}=0$ which implies that $W_0(\alpha\tau)=\mathsf{u}\,{\ii}$ for some non-zero $\mathsf{u}\in(-\pi,\pi)$. Following the definition of the Lambert $W$ function, then, we can write
\begin{align*}
\mathsf{u}\,\ee^{\mathsf{u}{\ii}}=(\re{\alpha}\tau+{\ii}\im{\alpha}\tau)\Leftrightarrow
~~\begin{cases}
-\mathsf{u}\sin(\mathsf{u})=\re{\alpha}\tau,\\ 
\,\,\,\,\mathsf{u}\cos(\mathsf{u})=\im{\alpha}\tau,
\end{cases}
\end{align*}
\begin{align*}
\Leftrightarrow
~~\begin{cases}
\mathsf{u}^2=\tau^2|\alpha|^2,\\
\tan(\mathsf{u})=\frac{\re{\alpha}}{\im{\alpha}},
\end{cases}
\end{align*}
which for $\mathsf{u}\in(-\pi,\pi)$, 
after eliminating $\mathsf{u}$ gives  $\tau={|\text{atan}(\frac{\re{\alpha}}{\im{\alpha}})|}\big\slash{|\alpha|}$ as the unique solution for $g(\alpha\tau)=0$. For $\alpha\in\real_{<0}$, we have $\text{atan}(\frac{\re{\alpha}}{\im{\alpha}})=\frac{\pi}{2}$, which means that $\bar{\tau}=\frac{\pi}{2|\alpha|}$. 
Finally, to validate~\eqref{eq::tau_tilde_real} we proceed as follows. 
$g(\alpha\tau)=1$ means that $\re{W_0(\alpha\tau)}=\re{\alpha\tau}$. Then, 
$W_0(\alpha\tau)=\re{\alpha\tau}+{\ii} \,\theta$ for some non-zero $\mathsf{\theta}\in(-\pi,\pi)$. Then, we obtain the value of $\theta$ as follows.  Invoking definition of Lambert function, we have
    \begin{align*}
   (\re{\alpha\tau}+\theta\ii)\ee^{\re{\alpha\tau}+\theta\ii}=\re{\alpha\tau}, 
    \end{align*}
    \begin{align*}
\Leftrightarrow
~~\begin{cases}
    \theta\cos{\theta}+\re{\alpha\tau}\sin{\theta}&=0,\\
    \ee^{\re{\alpha\tau}}\big(\re{\alpha\tau}\cos{\theta}-\theta\sin{\theta}\big)&=\re{\alpha\tau},
    \end{cases}
    \end{align*}
    which using some trigonometric manipulations can also be stated equivalently as 
    \begin{subequations}
    \begin{align}
      &  \re{\alpha\tau}=-\theta\cot(\theta)\label{eq::tilde_defining}\\
      &   \ee^{-\theta\cot(\theta)}=\cos(\theta). \label{eq::beta}
    \end{align}
    \end{subequations}
    For $\theta\in(-\pi,\pi)$, ~\eqref{eq::beta} has two distinct solutions $\theta\approx \pm1.01125$.
    Thus, the proof of~\eqref{eq::tau_tilde_real} follows from~\eqref{eq::tilde_defining}.
\end{proof}
To prove the rest of the results in Section~\ref{sec::Prob_formu} we rely on studying the derivative of $g(x)$ along $x=\alpha\tau$ with respect to $\tau\in\real_{>0}$ for a given $\alpha\in\complex$. Using \eqref{eq::lambert_derivative} and \eqref{eq::gain_def}, the derivative of delay rate gain function  along $x=\alpha \tau \neq- \frac{1}{\ee}$ with respect to time delay $\tau\in\real_{>0}$ can be written as
\begin{align*}
    \frac{\text{d}\,g(\alpha\tau)}{\text{d}\tau}=&\frac{1}{\re{\alpha}}\big[-\frac{1}{\tau^2}\re{W_0(\alpha\tau)}\!+\!\frac{1}{\tau}\re{\frac{\alpha}{\alpha\tau\!+\!\ee^{W_0(\alpha\tau)}}}\big],
\end{align*}
which can also be represented as 
\begin{align*}
\frac{\text{d}\,g(\alpha\tau)}{\text{d}\tau}&=-\frac{1}{\re{\alpha}\tau^2}\re{\frac{\alpha\tau\,W_0(\alpha\tau)}{\alpha\tau+\ee^{W_0(\alpha\tau)}}}\\&=-\frac{1}{\re{\alpha}\tau^2}\re{\frac{W_0^2(\alpha\tau)}{1+W_0(\alpha\tau)}}.
\end{align*}
Let $W_0(\alpha\tau)=\mathsf{w}+{\ii}\,\mathsf{u}$, where $\mathsf{u}\in(-\pi,\pi)$. Then, for $x=\alpha \tau \neq- \frac{1}{\ee}$, we can write
\begin{align}\label{eq::dot-gx}
 \frac{\text{d}\,g(\alpha\tau)}{\text{d}\tau}=&-\frac{1}{\re{\alpha}\tau^2}{\frac{\mathsf{w}^3+\mathsf{w}^2-\mathsf{u}^2+\mathsf{w}\mathsf{u}^2}{(\mathsf{w}+1)^2+\mathsf{u}^2}}\nonumber\\
 =&-\frac{1}{\re{\alpha}\tau^2}\,\frac{(\mathsf{w}^2+\mathsf{u}^2)\,\mathsf{w}+(\mathsf{w}^2-\mathsf{u}^2)}{(\mathsf{w}+1)^2+\mathsf{u}^2}.
\end{align}
At $\alpha\tau\!=\!-\frac{1}{\ee}$, the right and the left derivative of $g(\alpha\tau)$~along $x=\alpha\tau$ are obtained as follows. 
For $x\!=\!\alpha\tau\in[-\frac{1}{\ee},0]$, $W_0(x)\!\in\!\real$. Thus, by setting $\mathsf{u}=0$,~\eqref{eq::dot-gx} gives
\begin{align}\label{eq::rate_scalar_real_pos}
\frac{\text{d}\,g(\alpha\tau)}{\text{d}\tau}=\frac{1}{|\alpha|\,\tau^2}\,\frac{\mathsf{w}^2}{(\mathsf{w}+1)},~\quad \tau\in(0,\frac{1}{\ee|\alpha|}),
\end{align}
thus
\begin{align}\label{eq::rate_g_scalar_right_star}
&\lim_{\tau\to\frac{1}{\ee|\alpha|}\,\!^-}\!\!\frac{\text{d}\,g(\alpha\,\tau)}{\text{d}\tau}=\lim_{\mathsf{w}\to-1^+}\ee^2|\alpha|\,\frac{\mathsf{w}^2}{(\mathsf{w}+1)}=+\infty.
\end{align}
For any $x=\alpha\tau\in(-\infty,-\frac{1}{\ee})$, $W_0(x)=\mathsf{w}+{\ii}\,\mathsf{u}$ is a complex number with $\mathsf{u}\in(0,\pi)$, and satisfies
\begin{align*}
        (\mathsf{w}+{\ii}\,\mathsf{u})\ee^{\mathsf{w}+\mathsf{u}\,{\ii}}=\alpha\tau\Leftrightarrow\begin{cases}\ee^{\mathsf{w}}\,(\mathsf{w}\cos(\mathsf{u})-\mathsf{u}\sin(\mathsf{u}))=\alpha\tau,\\
        \ee^{\mathsf{w}}\,(\mathsf{u}\cos(\mathsf{u})+\mathsf{w}\sin(\mathsf{u}))=0.
        \end{cases}
\end{align*}

Therefore, for $\alpha\tau\in(-\infty,-\frac{1}{\ee})$, for which we always have  $\mathsf{u}\neq0$, we have $\mathsf{w}=-\mathsf{u}\cos(\mathsf{u})/\sin(\mathsf{u})$ and
\begin{align}\label{eq::rate_g_scalar_left_star}
&\frac{\text{d}\,g(\alpha\,\tau)}{\text{d}\tau}=\frac{1}{|\alpha|\,\tau^2}\,\frac{\mathsf{u}^2(-\mathsf{u}\frac{\cos(\mathsf{u})}{\sin(\mathsf{u})}+\cos(2\,\mathsf{u}))}{(-\mathsf{u}\,\cos(\mathsf{u})\!+\!\sin(\mathsf{u}))^2+\mathsf{u}^2\,\sin^2(\mathsf{u})}.
\end{align}
Using the L'Hospital's rule \cite[Theorem 5.5.2]{DC:05}, we can then colcude that 
\begin{align}\label{eq::rate_g_scalar_left_star}
&\lim_{\tau\to\frac{1}{\ee|\alpha|}\,\!^+}\!\!\frac{\text{d}\,g(\alpha\,\tau)}{\text{d}\tau}=\\
&~\ee^2|\alpha|~\lim_{\mathsf{u}\to0}\,\frac{\frac{\mathsf{u}^2}{\sin^2(\mathsf{u})}(-\mathsf{u}\frac{\cos(\mathsf{u})}{\sin(\mathsf{u})}+\cos(2\mathsf{u}))}{(-\mathsf{u}\frac{\cos(\mathsf{u})}{\sin(\mathsf{u})}+1)^2+\mathsf{u}^2}=-\frac{5\,\ee^2\,|\alpha|}{3}.\nonumber
\end{align}

Next is an intermediate result that we use in the proof Lemma~\ref{lem::g_0} and Lemma~\ref{lem::desired_region_complex}. To establish this result, we rely on the Jordan Curve Theorem, which states that a simple and closed curve divides the plane into an ``interior" region bounded by the curve and an ``exterior" region containing all of the nearby and far away exterior points~\cite{GB-WJ-RM-FR:75}.

\begin{lem}[Some of the properties of  level set $\mathcal{C}_0$ and superlevel set $\mathcal{S}_0$]\label{lem::C0S0_g(x)}
Consider the level set $\mathcal{C}_0$~\eqref{eq::C_c} and the superlevel set $\mathcal{S}_0$~\eqref{eq::S_c}. Let $\bar{\mathcal{C}}_0=\mathcal{C}_0\cup\{(0,0)\}$ and $\bar{\mathcal{S}}_0=\mathcal{S}_0\cup\{(0,0)\}$. Then, the following assertions hold.
\begin{itemize}
 \item[(a)]  $\bar{\mathcal{C}}_0$ is a simple closed curve in $\real^2$ that is symmetric about the $\mathsf{x}$ axis and intersects the $\mathsf{x}$ axis at only two points $\mathsf{x}=0$ and $\mathsf{x}=-\frac{\pi}{2}$. Moreover, it passes through the origin tangent to the $\mathsf{y}$ axis.
 
 \item[(b)] $\bar{\mathcal{S}}_0=\bar{\mathcal{C}}_0\cup\intor(\bar{\mathcal{C}}_0)$, and is a compact convex subset of $\real^2$.
 
 \item[(c)] $\mathcal{C}_c\subset\intor(\bar{\mathcal{C}}_0)\subset\bar{\mathcal{S}}_0$ for $c>0$ and $\mathcal{C}_c\subset\extor(\bar{\mathcal{C}}_0)$ for $c<0$.
\end{itemize}
\end{lem}
\begin{proof}
For a $c\in\real$, by definition, $g(x)=c$ for any $x=(\mathsf{x}+\mathsf{y}\,{\ii})\in\complex_{-}^l$ is equivalent to $\re{W_0(\mathsf{x}+\mathsf{y}{\ii})}=c\,\mathsf{x}$.
In other words, $W_0(\mathsf{x}+\mathsf{y}\,{\ii})=c\,\mathsf{x}+\mathsf{u}\,{\ii}$, where $\mathsf{u}\in(-\pi,\pi)$. Thereby, given property~\eqref{eq::W_0_conjugate_prop} of the Lambert $W_0$ function, each level set $\mathcal{C}_c$, $c\in\real$, is symmetric about the real axis.
Moreover, from the definition of the Lambert $W$ function we can~write 
\begin{align}\label{eq::C_c_rep}
(c\,\mathsf{x}+\mathsf{u}\,{\ii})\ee^{c\,\mathsf{x}+\mathsf{u}{\ii}}=\,\mathsf{x}+\mathsf{y}{\ii}\Leftrightarrow\!\begin{cases}\ee^{c\,\mathsf{x}}\,(c\,\mathsf{x}\cos(\mathsf{u})-\mathsf{u}\sin(\mathsf{u}))=\mathsf{x},\\
\ee^{c\,\mathsf{x}}\,(\mathsf{u}\cos(\mathsf{u})+c\,\mathsf{x}\sin(\mathsf{u}))=\mathsf{y}.
\end{cases}        
\end{align}
For $c=0$, by eliminating $\mathsf{u}$ in~\eqref{eq::C_c_rep} via trigonometric manipulations, we can characterize ${\mathcal{C}}_0$ by
\begin{align*}
      {\mathcal{C}}_0\!=\!\Big\{(\mathsf{x},\mathsf{y})\in\real_{<0}\!\times\!\real&\,\Big|\,
\frac{\mathsf{x}}{\mathsf{y}}=\pm\,\text{tan}\big(\sqrt{\mathsf{x}^2\!+\!\mathsf{y^2}}\,\big),\\
&\quad\quad\quad\quad0<(\mathsf{x}^2\!+\!\mathsf{y^2})<\pi^2 \Big\}.
\end{align*}
In polar coordinates, any $(\mathsf{x},\mathsf{y})\in{\mathcal{C}}_0$, reads as $\mathsf{x}=r\cos(\theta)$ and $\mathsf{y}=r\sin(\theta)$ with 
\begin{subequations}\label{eq::contour_g(x)=c_polar}
\begin{align*}
\begin{cases}
 \tan(r)=\pm\cot(\theta),&\theta\in(\frac{\pi}{2},\frac{3\pi}{2}),
 \\
 0< r<\pi,&\end{cases}
\end{align*} 
\end{subequations}
Therefore, ${\mathcal{C}}_0=\Big\{(\mathsf{x},\mathsf{y})\in\real_{<0}\!\times\!\real\,\Big|\,\mathsf{x}=r\cos(\theta), \mathsf{y}=\pm r\sin(\theta),~r=\,\theta- \frac{\pi}{2},~\theta\in(\pi/2,\pi]\Big\}.$ Evidently, 
\begin{align}\label{eq::C_0}
 \bar{\mathcal{C}}_0\!=\!\Big\{(\mathsf{x},\mathsf{y})\in\real_{\leq}0\!\times\!\real\,\Big|\,&\mathsf{x}=r\cos(\theta),~ \mathsf{y}=\pm r\sin(\theta),\nonumber \\
        &r=\,\theta- \frac{\pi}{2},~\theta\in
        [\pi/2,\pi]\Big\}.
       \end{align}
      From~\eqref{eq::C_0}, for any point on the upper half (resp. positive $\mathsf{y}$) and lower half (resp. negative $\mathsf{y}$) of $\bar{\mathcal{C}}_0$, in polar coordinates, $r$ is a continuous and a bounded function of $\theta$ and $\frac{d\,r}{d\,\theta}=1$ exists on $\theta\in(\pi/2,\pi)$ (resp. $\theta\in(\pi,3\pi/2)$). Therefore, $\bar{\mathcal{C}_0}$ is a simple closed curve. Next, note that on $\bar{\mathcal{C}}_0$, as $\theta\rightarrow\frac{\pi^+}{2}$, and due to symmetry also as $\theta\rightarrow\frac{3\pi^-}{2}$, it follows that $r\rightarrow0$. Therefore, $\bar{\mathcal{C}}_0$ passes through the origin tangent to the $\mathsf{y}$ axis. Also, at $\theta=\frac{\pi}{2}$ and $\theta=\pi$ we have, respectively,  $(\mathsf{x},\mathsf{y})=(0,0)$ and $(\mathsf{x},\mathsf{y})=(-\frac{\pi}{2},0)$. 
       As a result, assertion (a) holds. 
      
      Since $\bar{\mathcal{C}}_0$ is a simple closed curve, it follows from the Jordan Curve Theorem that $\bar{\mathcal{C}}_0$ divides the plane into an interior region bounded by ${\mathcal{C}}_0$ and an exterior region containing all of the nearby and far away exterior points.  Moreover, note that  
       the curvature of the upper half of $\bar{\mathcal{C}}_0$ in a $(\mathsf{x},\mathsf{y})$ plane is $\kappa=\frac{r^2+2r_\theta^2-rr_{\theta\theta}}{(r^2+r_\theta^2)^{\frac{3}{2}}}=\frac{r^2+2}{(r^2+1)^\frac{3}{2}}>0$, where $(.)_\theta=\partial(.)/\partial \theta$~\cite{HSMC:89}. Therefore the upper half curve of $\mathcal{C}_0$ is a convex curve.
       Consequently, $\mathcal{Z}=\bar{\mathcal{C}_0}\cup\text{int}(\bar{\mathcal{C}}_0)$ is a compact convex set (recall that $\mathcal{C}_0$ is symmetric about $\mathsf{x}$ axis). To complete the proof of the statement (b), we show that $\bar{{\mathcal{S}_0}}=\mathcal{Z}$.

       Consider an $\alpha\in\complex^l_{-}$ and recall $\bar{\tau}$ in~\eqref{eq::tau_bar}. Because $(0,0)\in\text{bd}(\mathcal{Z})$ and $(\re{\alpha\bar{\tau}},\im{\alpha\bar{\tau}})\in\text{bd}(\mathcal{Z})$ (recall $g(\alpha\bar{\tau})=0$), it follows from the compact convexity of $\mathcal{Z}$ that  $(\re{\alpha\tau},\im{\alpha\tau})\in\intor(\mathcal{C}_0)$ for $\tau\in(0,\bar{\tau})$ and $(\re{\alpha\bar{\tau}},\im{\alpha\bar{\tau}})\in\extor(\mathcal{C}_0)$ for $\tau\in(\bar{\tau},\infty)$.
       Therefore, given that $g(\alpha\tau)$ is a continuous function of $\tau\in\real_{\geq0}$ (see Lemma~\ref{lem::g_continuous}) and $g(\alpha\tau)$ at $\tau=0$ and $\tau=\bar{\tau}$ is equal to, respectively, $1$ and $0$, we can conclude that $g(\alpha\tau)>0$ for any $\tau\in(0,\bar{\tau})$. 
       This means that at any $(\re{\alpha\tau},\im{\alpha\tau})\in\intor(\mathcal{C}_0)$, we have $g(\alpha\tau)>0$. Consequently, $\mathcal{Z}=\mathcal{S}_0$, which completes the proof of statement (b).
       
       From validity of the statement (b), we can readily deduce that $\mathcal{C}_c\subset\intor(\mathcal{C}_0)\subset\mathcal{S}_0$ for $c>0$.  To complete the proof of the statement (c), 
       we recall from the proof of the statement (b) that for any $\alpha\in\complex^{l}_{-}$, $(\re{\alpha{\tau}},\im{\alpha{\tau}})\in\extor(\mathcal{C}_0)$ for $\tau\in(\bar{\tau},\infty)$, i.e., $g(\alpha\tau)\neq0$ for $\tau\in(\bar{\tau},\infty)$. Then, combined with the fact that $g(\alpha\tau)$ is a continuous function of $\tau\in\real_{\geq0}$  and also that at $\bar{\tau}$ from~\eqref{eq::dot-gx} we have (recall that $\re{\alpha}<0$)
       \begin{align}\label{eq::dg_bar_tau}
           \frac{\text{d}g(\alpha\bar{\tau})}{\text{d}\tau}=\frac{1}{\re{\alpha}\bar{\tau}^2}\,\frac{\bar{\mathsf{u}}^2}{1+\bar{\mathsf{u}}^2}<0,
       \end{align}
       we can conclude that $g(\alpha\tau)<0$ for $\tau\in(\bar{\tau},\infty)$. This means that $\mathcal{C}_c\subset\extor(\bar{\mathcal{C}}_0)$ for $c<0$. To arrive at~\eqref{eq::dg_bar_tau}, we relied on the knowledge that $g(\alpha\bar{\tau})=0$ indicates that $\re{W_0(\alpha\bar{\tau})}=0$, therefore $W_0(\alpha\bar{\tau})=0+\bar{\mathsf{u}}\ii$ for some $\bar{\mathsf{u}}\in(-\pi,\pi)$. 
\end{proof}
The topological properties of the the $\mathcal{C}_0$ level set, which are established in Lemma~\ref{lem::C0S0_g(x)} is evident in Fig.~\eqref{Fig::g_countor}. Using the results of Lemma~\ref{lem::C0S0_g(x)}, we proceed next to present the proof of Lemma~\ref{lem::g_0}.
\begin{proof}[Proof of Lemma~\ref{lem::g_0}]
 Recall from the statement (b) of Lemma~\ref{lem::C0S0_g(x)} that $\mathcal{S}_0$ is a compact convex set. Therefore, since  $(\re{\alpha\bar{\tau}},\im{\alpha\bar{\tau}})\in\text{bd}(\bar{\mathcal{S}}_0)$ (recall $g(\alpha\bar{\tau})=0$) and $(0,0)\in\text{bd}(\bar{\mathcal{S}}_0)$, then $(\re{\alpha\tau},\im{\alpha\tau})\in\intor(\bar{\mathcal{C}}_0)$ for $\tau\in(0,\bar{\tau})$ and $(\re{\alpha\bar{\tau}},\im{\alpha\bar{\tau}})\in\extor(\bar{\mathcal{C}}_0)$ for $\tau\in(\bar{\tau},\infty)$. Then, the proof follows from the statement (c) of Lemma~\ref{lem::C0S0_g(x)}.
\end{proof}
The proof of Lemma~\ref{lem::g_0} can also be deduced from the \emph{continuity stability property} theorem~\cite[Proposition 3.1]{SN:01}
for linear delayed systems. In this regard, consider the dynamical system $\dot{x}=\left[\begin{smallmatrix}\re{\alpha}&\im{\alpha}\\
       -\im{\alpha}&\re{\alpha}\end{smallmatrix}\right]x(t-\tau)$, whose eigenvalues are $\alpha$ and $\conj(\alpha)$. The real part of the RMR of the CE of this system is given by $\re{S_\tau^r}=g(\alpha\tau)\re{\alpha}$ (recall~\eqref{eq::gain_def}).
      It follows from the \emph{continuity stability property} theorem~\cite[Proposition 3.1]{SN:01} and Lemma~\ref{thm::admis_tau} that $\re{s_\tau^r}\in\real_{<0}$ if and only $\tau\in[0,\bar{\tau})$.  
      Therefore, $g(\alpha\tau)> 0$ if and only if  $\tau\in[0,\bar{\tau})$, which along with the fact  $g(\alpha\bar{\tau})=0$ validates the statement of Lemma~\ref{lem::g_0}.

Next, we prove Lemma~\ref{lem::alpha-real-positiverate}.
\begin{proof}[Proof of Lemma~\ref{lem::alpha-real-positiverate}]
At $\tau^\star=1/(\ee|\alpha|)$, because $W_0(-\ee^{-1})=-1$, we have   $g(\alpha\tau^\star)=\re{W_0(-\ee^{-1})}/(-\ee^{-1})=\ee$.
Next, note that from~\eqref{eq::lim-g-0},~\eqref{eq::tau_bar} and Lemma~\ref{lem::g_continuous}, we know, respectively, that $\lim_{\tau\to0}g(\alpha\,\tau)=1$, $\tau=\bar{\tau}$ is the unique solution of $g(\alpha\tau)=0$ for $\tau\in\real_{>0}$, and $g(\alpha\tau)$ is a continuous function of $\tau\in\real_{>0}$.
Therefore, to complete the proof of statement (a), we show next that $\frac{d\,g(x)}{d\tau}>0$ for $\tau\in(0,\tau^\star)$ and $\frac{d\,g(x)}{d\tau}<0$ for $\tau\in(\tau^\star,\bar{\tau})$.
For $x=\alpha\tau\in[-\frac{1}{\ee},0]$, we have $W_0(x)\in\real$, and as such by setting $\mathsf{u}=0$ from~\eqref{eq::dot-gx} 
we obtain 
\begin{align}\label{eq::rate_scalar_real_pos}
\frac{\text{d}\,g(\alpha\tau)}{\text{d}\tau}=\frac{1}{|\alpha|\,\tau^2}\,\frac{\mathsf{w}^2}{(\mathsf{w}+1)}>0,~\quad \tau\in(0,\tau^\star).
\end{align}
Moreover, from~\eqref{eq::rate_g_scalar_right_star} we have
\begin{align}\label{eq::rate_g_scalar_right_star2}
&\lim_{\tau\to\tau^\star\,\!^-}\!\!\frac{\text{d}\,g(\alpha\,\tau)}{\text{d}\tau}=+\infty.
\end{align}
For any $x=\alpha\tau\in(-\infty,-\frac{1}{\ee})$, recall that $\frac{\text{d}\,g(\alpha\,\tau)}{\text{d}\tau}$ is given by~\eqref{eq::rate_g_scalar_left_star}. Because for $\alpha\tau\in[\frac{-\pi}{2},-\frac{1}{\ee})$ we have $\mathsf{u}\in(0,\frac{\pi}{2}]$, we can confirm that 
$(-\mathsf{u}\frac{\cos(\mathsf{u})}{\sin(\mathsf{u})}+\cos(2\,\mathsf{u}))<0$ and therefore, we obtain
\begin{align}\label{eq::rate_scalar_real_neg}
&\frac{\text{d}\,g(\alpha\,\tau)}{\text{d}\tau}<0,~\quad \tau\in(\tau^\star,\bar{\tau}].
\end{align}
Next, note that from~\eqref{eq::rate_g_scalar_left_star} we have 
\begin{align}\label{eq::rate_g_scalar_left_star2}
&\lim_{\tau\to\tau^\star\,\!^+}\!\!\frac{\text{d}\,g(\alpha\,\tau)}{\text{d}\tau}=-\frac{5\,\ee^2\,|\alpha|}{3}<0.
\end{align}
In light of the observations above, the proof of the statement (b) follows from~\eqref{eq::lim-g-0},~\eqref{eq::tau_tilde_real},~\eqref{eq::rate_scalar_real_pos},~\eqref{eq::rate_scalar_real_neg} and the continuity of $g(\alpha\tau)$ for $\tau\in\real_{>0}$. Finally, statement (c) is deduced from statements (a) and (b), along with~\eqref{eq::rate_g_scalar_right_star2} and~\eqref{eq::rate_g_scalar_left_star2}. Note that since $g(\alpha\tau^\star)=\ee>1$, we have $\tau^\star\in[0,\tilde{\tau})$.
\end{proof}

\medskip

Next, we use the results of Lemma~\ref{lem::C0S0_g(x)} to establish the proof of Lemma~\ref{lem::desired_region_complex}, which characterizes the variation of $\frac{d\,g(x)}{d\tau}$ along $x=\alpha\tau$ for $\tau\in\real_{>0}$ such that $(\re{x},\im{x})\in\mathcal{S}_0$. 
\begin{proof}[Proof of Lemma~\ref{lem::desired_region_complex}]
By virtue of Lemma~\eqref{lem::g_0}, we know that $g(\alpha\tau)>0$ for $\tau\in(0,\bar{\tau})$, and $g(\alpha\bar{\tau})=0$.  Then, it follows from the definition of delay rate gain and  $\re{\alpha}<0$ that  $\re{W_0(\alpha\tau)}<0$ for $\tau\in(0,\bar{\tau})$.
Consequently, for $\tau\in(0,\bar{\tau}]$, we have $(\re{\alpha\tau},\im{\alpha\tau})\in\mathcal{S}_0$ and $(\re{W_0(\alpha\tau)},\im{W_0(\alpha\tau)}\in\real_{\leq0}\times(-\pi,\pi)$.

Next, recall that the derivative of $g(x)$ with respect to $\tau$ along $x=\alpha\tau$, except at $x=-\frac{1}{\ee}$, is given by~\eqref{eq::dot-gx}, where we denoted $W_0(\alpha\tau)=\mathsf{w}+\mathsf{u}\,{\ii}$, which satisfies $\mathsf{u}\in(-\pi,\pi)$.  Since $\re{\alpha}<0$, it is perceived from~\eqref{eq::dot-gx} that the sign of $\frac{d\,g(x)}{d\tau}$ along $x=\alpha\tau$ is defined solely by the sign of $\psi=(\mathsf{w}^2+\mathsf{u}^2)\,\mathsf{w}+(\mathsf{w}^2-\mathsf{u}^2)$, the nominator of~\eqref{eq::dot-gx}. 

Using the polar coordinates and a set of simple trigonometric manipulations, we identity the points $(\mathsf{w},\mathsf{u})\in\real_{\leq0}\times(-\pi,\pi)$ for which $\psi$ retains a zero, a positive and a negative value, as, respectively, $\Gamma$, $\Gamma^{+}$ and $\Gamma^{-}$, where 
\begin{align*}
{\Gamma}=\,&\Big\{(\mathsf{w},\mathsf{u})\in\,\real_{\leq0}\times(-\pi,\pi)\,\big|\,\mathsf{w}=\mathsf{R}\cos(\theta),~\mathsf{u}=\mathsf{R}\sin(\theta),\nonumber\\
&\qquad\quad 3\pi/4\leq\theta\leq5\pi/4 \,,\, \mathsf{R}=-\cos(2\theta)/\cos(\theta)\Big\},\\
{\Gamma}^{+}=\,&\Big\{(\mathsf{w},\mathsf{u})\in\real_{\leq0}\times(-\pi,\pi)\,\big|\,\mathsf{w}=\mathsf{R}\cos(\theta),~\mathsf{u}=\mathsf{R}\sin(\theta),\nonumber\\
&\qquad\quad\quad 3\pi/4\leq\theta\leq5\pi/4 \,,\, \mathsf{R}<-\cos(2\theta)/\cos(\theta)\Big\},\\
{\Gamma}^{-}=\,&\left(\real_{\leq0}\times(-\pi,\pi)\right)\backslash(\Gamma\cup\Gamma^{+}).
\end{align*}
Therefore, for $x=\alpha\tau$ with $\tau\in(0,\bar{\tau}]$ we have 
\begin{align*}
    \frac{\text{d}\,g(x)}{\text{d}\tau}>0 &~~\text{if}~~ (\re{W_0(x)},\im{W_0(x)})\in\Gamma^{+},\\
    \frac{\text{d}\,g(x)}{\text{d}\tau}=0 &~~ \text{if}~~ (\re{W_0(x)},\im{W_0(x)})\in({\Gamma}\backslash\{(-1,0)\}),\\
    \frac{\text{d}\,g(x)}{\text{d}\tau}<0&~~ \text{if}~~ (\re{W_0(x)},\im{W_0(x)})\in\Gamma^{-}.
\end{align*}
Here we used $W_0(-\frac{1}{\ee})=-1$.


Now, let $\alpha\tau=r\,{\ee}^{\phi\ii}$ and $W_0(\alpha\tau)=\mathsf{R}\ee^{\theta\ii}=\mathsf{R}\cos(\theta)+\mathsf{R}\sin(\theta)\ii$, in which $-\pi<\mathsf{R}\sin(\theta)<\pi$ by definition. Then, using the relation $W_0(\alpha\tau)\,\ee^{W_0(\alpha\tau)}=\alpha\tau$ we obtain that
\begin{align*}
    r&=\,\mathsf{R}\,\ee^{\mathsf{R}\cos(\theta)},\label{eq::R}\\
    \phi&=\theta+\mathsf{R}\,\sin(\theta).
\end{align*}
Then, recalling the definition of $\Gamma$, $\Gamma^{+}$ and $\Gamma^{-}$, the proof of part (a) follows from the observation that for $x=\alpha\tau$, $\tau\in(0,\bar{\tau}]$  we have
\begin{alignat*}{3}
(\re{\alpha\tau},\im{\alpha\tau})&\in\intor(\Lambda)&\Rightarrow~~(\mathsf{w},\mathsf{u})\in\Gamma^{+},\\
(\re{\alpha\tau},\im{\alpha\tau})&\in\Lambda&\Rightarrow~~(\mathsf{w},\mathsf{u})\in\Gamma~~,\\
(\re{\alpha\tau},\im{\alpha\tau})&\in(\mathcal{S}_0\backslash(\intor(\Lambda)\cup\Lambda))~~&\Rightarrow~~(\mathsf{w},\mathsf{u})\in\Gamma^{-},
\end{alignat*}
where $(\mathsf{w},\mathsf{u})=(\re{W_0(\alpha\tau)},\im{W_0(\alpha\tau)})$. Recall here that $W_0$ function is injective.

To prove statement (b) and (c) we proceed as follows. We note that in~\eqref{thm::desired_region_complex} $r$ and $\phi$ both are a differentiable and a bounded function of $\theta\in[\frac{3\pi}{4},\frac{5\pi}{4}]$. Moreover, 
\begin{align}
\frac{\text{d}\,\phi}{\text{d}\,\theta}=&1+2\tan(\theta)\,\ee^{-\cos{2\theta}}\sin(2\theta)\,(1-\cos(2\theta))-\nonumber\\&\cos(2\theta)\,\ee^{-\cos(2\theta)}\,(1+\tan^2(\theta))>0,
\end{align}

for any $\theta\in[\frac{3\pi}{4},\frac{5\pi}{4}]$.  Therefore, $\phi$ in~\eqref{thm::desired_region_complex} has a one-to-one correspondence with $\theta\in[\frac{3\pi}{4},\frac{5\pi}{4}]$, which in turn indicates that $r$ in~\eqref{thm::desired_region_complex} has a one-to-one correspondence with $\phi$. Moreover, since  $\phi(\frac{3\pi}{4})=\frac{3\pi}{4}$ and $\phi(\frac{5\pi}{4})=\frac{5\pi}{4}$, we have $\phi\in[\frac{3\pi}{4},\frac{5\pi}{4}]$. In addition, we can also conclude that $\frac{\text{d}r}{\text{d}\phi}=\frac{\text{d}r}{\text{d}\theta}/\frac{\text{d}\phi}{\text{d}\theta}$ exists and is finite at every $\theta\in[\frac{3\pi}{4},\frac{5\pi}{4}]$. Combined with $r$ satisfying $r(\frac{3\pi}{4})=r(\frac{5\pi}{4})=0$, we can then conclude that $\Lambda$ is a simple closed curve. Therefore, it follows from the Jordan Curve theorem that $\intor(\Lambda)\cup\Lambda$ is a connected compact subset of $\real^2$. In light of the preceding observations, we  make the following conclusions. 
The ray $(\re{\alpha\tau},\im{\alpha\tau})$, $\tau\in(0,\bar{\tau}]$, intersects $\Lambda$ if and only if $\frac{3\pi}{4}<\text{arg}(\alpha)<\frac{5\pi}{4}$. 
Therefore, if $\text{arg}(\alpha)\not\in(\frac{3\pi}{4},\frac{5\pi}{4})$, we have $(\re{\alpha\tau},\im{\alpha\tau})\in(\mathcal{S}_0\backslash(\intor(\Lambda)\cup\Lambda))$  for $\tau\in(0,\bar{\tau}])$. Then, the proof of the statement (b) follows from the statement (a). 
On  the  other hand,  if  $\text{arg}(\alpha)\in(\frac{3\pi}{4},\frac{5\pi}{4})$,  then due to the one-to-one  correspondence between $r$ and $\phi$, ray  $(\re{\alpha\tau},\im{\alpha\tau})$, $\tau\in(0,\bar{\tau}]$ intersects $\Lambda$ at a unique point. Let this point correspond to $\tau^\star\in(0,\bar{\tau}]$, i.e., $(\re{\alpha\tau^\star},\im{\alpha\tau^\star})\in\Lambda$. Then,~\eqref{eq::tau-star-complex} is deduced from the definition of $\Lambda$ in~\eqref{thm::desired_region_complex}. Next, note that from compactness of $\intor(\Lambda)\cup\Lambda)$ and the fact that $(\re{\alpha\tau},\im{\alpha\tau})$, $\tau\in(0,\bar{\tau}]$, intersects $\Lambda$ at a unique point, it follows that $(\re{\alpha\tau},\im{\alpha\tau})\in\intor(\Lambda)$ for $\tau\in(0,\tau^\star)$ and $(\re{\alpha\tau},\im{\alpha\tau})\in(\mathcal{S}_0\backslash(\intor(\Lambda)\cup\Lambda))$. Thereby, by virtue of the statement (a) we conclude that $\frac{\text{d} g(\alpha\tau)}{\text{d}\tau}>0$ for $\tau\in(0,\tau^\star)$, and 
$\frac{\text{d} g(\alpha\tau)}{\text{d}\tau}<0$ $\tau\in(\tau^\star,\bar{\tau}]$. If $\alpha\in\real_{<0}$, then $\arg(\alpha)=\pi$. Therefore, $\tau^\star$ in~\eqref{eq::tau-star-complex} becomes equal to $\frac{1}{\ee|\alpha|}$, and consequently, ~\eqref{eq::drev_g_at_tau_star} follows from~\eqref{eq::rate_g_scalar_right_star} and~\eqref{eq::rate_g_scalar_left_star}. Lastly, when $\alpha\not\in\real_{<0}$, since  $(\re{\alpha\tau^\star},\im{\alpha\tau^\star})\in\Lambda\backslash\{(-\frac{1}{\ee},0)\}$, $\frac{\text{d} g(\alpha\tau)}{\text{d}\tau}=0$ at $\tau=\tau^\star$ is deduced from the statement (a). 
\end{proof}

Figure~\ref{Fig::stab_complex} depicts $\Lambda$ in~\eqref{thm::desired_region_complex} (red curve) in a $(\mathsf{x},\mathsf{y})$ plane along with the level sets $\mathcal{C}_1$ (green curve) and $\mathcal{C}_0$ (blue curve).  As one can expect (given~\eqref{eq::g0_bar_tau} and the statement (a) of Lemma~\ref{lem::desired_region_complex}), $\Lambda$ is located inside $\mathcal{C}_1$ and between the lines $\mathsf{y}=\pm\mathsf{x}$.
It is interesting to note that lines $\mathsf{y}=\pm\mathsf{x}$ are also tangent to $\mathcal{C}_1$ at the origin. This observation can be verified as follows. Consider $(\mathsf{x}_1,\mathsf{y}_1)\in\mathcal{C}_1$ and let $\mathsf{x}_1+\mathsf{y}_1\ii=\mathsf{r}\,\ee^{{\ii}\theta}$. Then, for $(\mathsf{x}_1,\mathsf{y}_1)$ in the close neighborhood of the origin from~\eqref{eq::lim-g-0}) we expect that $\lim_{r\to0}-\cos(2\,\theta)+\frac{3}{2}\mathsf{r}\cos(3\,\theta)-\frac{8}{3}\mathsf{r}^2\cos(4\,\theta)+\cdots=0  $. This limit is possible only if $\theta\to\frac{3\pi}{4}$ and $\theta\to\frac{5\pi}{4}$ (solution of $\cos(2\,\theta)=0$ for $\theta\in[\pi/2,3\pi/2]$). This verifies that as $(\mathsf{x}_1+\mathsf{y}_1\ii)\to0$ on $\mathcal{C}_1$, $\mathcal{C}_1$  become tangent to the lines $\mathsf{y}=\pm \mathsf{x}$.


\end{document}